\documentclass[12pt]{amsart}

\usepackage{amsmath,amsthm,amssymb,amsfonts,verbatim}
\usepackage[colorlinks,
            linkcolor=blue,
            anchorcolor=blue,
            citecolor=blue
            ]{hyperref}
\usepackage[hmargin=1.2in,vmargin=1.2in]{geometry}
\usepackage{array}
\usepackage{multirow}

\title[Upper Bounds on Maximum Lengths of Singleton-Optimal LRCs]{Upper Bounds on Maximum Lengths of Singleton-Optimal Locally Repairable Codes}
\author{Shu Liu}\address{National Key Laboratory of Science and Technology on Communications, University of Electronic Science and Technology of China, China} \email{shuliu@uestc.edu.cn}
\author{Tingyi Wu}\address{ Theory Lab, Central Research Institute, 2012 Labs, Huawei Technology Co. Ltd.} \email{wu.ting.yi@huawei.com}
\author{Chaoping Xing} \address{School of Electronic Information and Electric Engineering, Shanghai Jiao Tong University, China} \email{xingcp@sjtu.edu.cn}
\author{Chen Yuan} \address{School of Electronic Information and Electric Engineering, Shanghai Jiao Tong University, China} \email{{chen}${\_}${yuan}@sjtu.edu.cn}
\date{}

\newtheorem{lemma}{Lemma}[section]
\newtheorem{theorem}[lemma]{Theorem}
\newtheorem{cor}[lemma]{Corollary}

\newtheorem{prop}[lemma]{Proposition}
\newtheorem{ex}[lemma]{Example}

\newtheorem{defn}{Definition}
\newtheorem{exm}{Example}
\newtheorem{open}{Open Problem}

\theoremstyle{remark}
\newtheorem{rmk}{Remark}

\renewcommand{\epsilon}{\varepsilon}
\renewcommand{\le}{\leqslant}
\renewcommand{\ge}{\geqslant}

\newcommand{\vnote}[1]{}

\def\F{\mathbb{F}}
\def \mC {\mathcal{C}}

\def \mC {\mathcal{C}}

\def \Xi {{X^{[i]}}}

\newcommand{\Ga}{\alpha}
\newcommand{\Gb}{\beta}

\newcommand{\Ge}{\epsilon}

\newcommand{\Gl}{\lambda}

\def \bc {{\bf c}}

\def \bu {{\bf u}}
\def \bv {{\bf v}}
\def \bo {{\bf 0}}

\def\supp {{\rm supp }}

\def\LRC {{\rm locally repairable code\ }}
\def\LRCs {{\rm locally repairable codes\ }}

\begin{document}
\maketitle

\begin{abstract} A locally repairable code is called Singleton-optimal if it achieves the Singleton-type bound. Such codes are of great theoretic interest in the study of locally repairable codes. In the recent years there has been a great amount of work on this topic. One of the main problems in this topic is to determine the largest length of a $q$-ary Singleton-optimal locally repairable code for given locality and minimum distance. Unlike classical MDS codes, the maximum length of Singleton-optimal locally repairable codes are very sensitive to minimum distance and locality. Thus, it is more challenging and complicated to investigate  the maximum length of Singleton-optimal locally repairable codes.

In literature, there has been already some research on this problem. However, most of work is concerned with some specific parameter regime such as small minimum distance and locality, and  rely on the constraint that $(r+1)|n$ and recovery sets are disjoint, where $r$ is locality and $n$ is the code length. In this paper we study the problem for large range of parameters including the case where minimum distance is proportional to length. In addition, we also derive  some upper bounds on  the maximum length of Singleton-optimal locally repairable codes with small minimum distance by removing this constraint. It turns out that even without the constraint we still get better upper bounds for codes with small locality and distance compared with known results. Furthermore, based on our upper bounds for codes with small distance and locality and some propagation rule that we propose in this paper, we are able to derive some upper bounds for codes  with relatively large distance and locality assuming that $(r+1)|n$ and recovery sets are disjoint.
\end{abstract}

\section{Introduction}
Modern distributed storage systems have great demand for erasure coding based schemes with good storage efficiency in order to cope with the explosion in the amount of data stored online. Locally Repairable Codes (LRCs) have emerged as the codes of choice for many such scenarios and have been implemented in a number of large scale systems e.g., Microsoft Azure~\cite{HSX+12} and Hadoop~\cite{SAP+13}. The topic has attracted many researchers and a great amount of work has been done in the literature \cite{HL07,HCL,GHSY12,PKLK12,SRKV13,FY14,PD14,TB14,TPD16,BTV17,CMST20}.

A block code is called a locally repairable code  with locality $r$ if every symbol in the encoding is a function of
$r$ other symbols. This enables recovery of any single erased symbol in a local fashion by downloading at most $r$ other symbols. On the other hand, one would like the code to have a good minimum distance to enable recovery of many erasures in the worst-case. Locally repairable codes have been the subject of extensive study in recent years.
Locally repairable codes offer a good balance between efficient erasure recovery in the typical case in distributed storage systems where a single node fails (or becomes temporarily unavailable due to maintenance or other causes), and still allowing recovery of the data from a larger number of erasures and thus safeguarding the data in more worst-case scenarios.

Codes with information locality property were also studied in \cite{FY14,GHSY12}. In this paper we focus on linear codes. For an $[n,k,d]$-linear code with locality $r$, we denote it by $[n,k,d;r]$  (see the definition of locally repairable codes in Section 2).  An $[n,k,d;r]$-locally repairable code must obey various upper bounds such as the Singleton-type bound \cite{HSX+12}, the Cadambe-Mazumdar bound \cite{CM15}, etc. A locally repairable code achieving any of these upper bounds is called optimal. Among these bounds, the Singleton-type bound is neat and independent of alphabet size. Precisely speaking, for a linear locally repairable code ${\mC}$ of length $n$ with $k$ information symbols and locality $r$,    the minimum distance $d({\mC})$ of ${\mC}$ is upper bounded by
 \begin{equation}\label{eq:1}
 d({\mC})\le n-k-\left\lceil \frac kr\right\rceil+2.
 \end{equation}
 The bound \eqref{eq:1} is called the Singleton-type bound for locally repairable codes and it was first proved in \cite{GHSY12}. A linear code achieving the above Singleton-type bound is called Singleton-optimal.

 A code achieving  the classical Singleton bound is called a Maximal Distance Separable (MDS) code. There is a well-known Main MDS Conjecture stating that a $q$-ary nontrivial MDS code has length at most $q+2$. This conjecture is believed to be true widely and the conjecture has been proved when the ground field has a prime order \cite{B12}. The Main MDS Conjecture is very neat in the sense that the maximum length of a nontrivial MDS code depends only on $q$ and is independent of minimum distance. Unlike the Main MDS Conjecture, the maximum length of a Singleton-optimal locally repairable code is very sensitive to minimum distance and locality. This makes the problem more challenging and complicated.
  The current paper will focus on study of maximum lengths of  Singleton-optimal locally repairable codes, namely we will drive some upper bounds.

 \subsection{Known results}
One natural question is how large the maximum length of a Singleton-optimal locally repairable code could be. There is already some research on this problem though it is still far from complete. The problem was first studied in the  paper by Guruswami et al. \cite{GXY19} where they considered the cases of relatively small distance $d$ and locality $r$ under the assumption that $(r+1)|n$ and recovery sets are disjoint, where $n$ is length of the code.
\begin{prop}[see \cite{GXY19}]\label{prop:1.1}
\label{prop:GXY} Assume that ${\mC}$ is a $q$-ary Singleton-optimal $[n,k,d;r]$-locally repairable code with $d \ge 5$ and $n \ge \Omega(d r^2)$. If  $(r+1)|n$ and recovery sets are disjoint, then  $n \le O(d q^{3-\frac{4}{d-1}})$ when $d=1 \pmod 4$, $n \le O(d q^{3})$ when $d=2,3 \pmod 4$ and $n \le O(d q^{3+\frac{4}{d-4}})$ when $d=0\pmod 4$. In particular, $n=O(q^2)$ if $d=5$.
\end{prop}
The above result was further improved in \cite{XY18}. More precisely speaking, the authors of \cite{XY18} proved the following.
\begin{prop}[see \cite{XY18}]\label{prop:1.2}
Let ${\mC}$ be a $q$-ary   Singleton-optimal $[n,k,d;r]$-locally repairable code with $d \ge 5$ and $n \ge \Omega(d r^2)$. If  $(r+1)|n$ and recovery sets are disjoint, then
\begin{equation}\label{eq:2}
n\leq\left\{ \begin{array}{ll}
O(q^3) &\mbox{if $d\pmod {r+1}> 5$ or $d\pmod {r+1}<2$},\\
O(q^2)&\mbox{if $2\leq d\pmod {r+1}\leq 5$.}\end{array}\right.
\end{equation}
\end{prop}
Recently, Chen et al.\cite{CFXHF} derived some upper bounds for small distance and locality such as $d=5,6$ and $r=1,2,3$. Let us quote their result below for an easy reference.
 \begin{prop}[see \cite{CFXHF}]\label{prop:1.3}
Let ${\mC}$ be a $q$-ary   Singleton-optimal $[n,k,d;r]$-locally repairable code with the constraint that $(r+1)|n$ and recovery sets are disjoint, then
\begin{itemize}
\item[{\rm (i)}] for $d=5$, one has $n=O(q^2)$. Furthermore, $n=O(q)$ if $r=1,2$;
\item[{\rm (ii)}] for $d=6$, one has $n=O(q^3)$. Furthermore, $n=O(q^{1.5})$ if $r=2$ and $n=O(q^{2})$ if $r=3$.
\end{itemize}
\end{prop}

Apart from the above bounds, there are various constructions of Singleton-optimal locally repairable codes in the literature. A class of codes constructed earlier and known as pyramid codes \cite{HCL} are shown to be codes that
are Singleton-optimal.  In \cite{SRKV13},  Silberstein et al. proposed a two-level construction based on the  Gabidulin codes combined with a single parity-check $(r+1,r)$ code of length $r+1$ and $r$ information symbols. Another construction \cite{TPD16} used two layers of MDS codes, a Reed-Solomon code and a special $(r+1,r)$ MDS code. A common shortcoming of these constructions relates to the size of the code alphabet which in all the papers is an exponential function of the code length, complicating the implementation. There  was an earlier construction of optimal locally repairable codes given in \cite{PKLK12} with  alphabet  size comparable to code length. However, the construction in \cite{PKLK12} only produces  a specific value of the length $n$, i.e., $n=\left\lceil \frac kr\right\rceil(r+1)$. Thus, the rate of the code is very close to $1$.

  The first breakthrough construction was given  in \cite{TB14}. This construction naturally generalizes Reed-Solomon construction which relies on the alphabet of cardinality comparable to the code length $n$. The idea behind the construction is very nice. The only shortcoming of this construction is restriction on locality $r$. Namely,  $r+1$ must be a divisor of either $q-1$ or $q$, or $r+1$ is equal to a product of a divisor of $q-1$ and a divisor of $q$ for certain $q$, where $q$ is the code alphabet. There are also some existence results given in \cite{PKLK12} and \cite{TB14} with less restriction on locality $r$. But both results require large alphabet which is an exponential function of the code length.
This construction was extended via automorphism group of rational function fields by Jin, Ma and Xing \cite{JMX19} and it turns out that there are more flexibility on locality and the code length can be $q+1$.

Recently, there are also some constructions of Singleton-optimal locally repairable codes. For instance, the construction in \cite{Jin19} provides Singleton-optimal locally repairable codes with distance $d=5$ and $6$. By using elliptic curves, some Singleton-optimal locally repairable codes with distance proportional to length are constructed in \cite{LMX20}. For $(d,r)=(6,2)$, three Singleton-optimal locally repairable codes with length $\Omega(q)$ are constructed in \cite{CFXHF}. In addition, a class of Singleton-optimal locally repairable codes with a large range of distance and locality are presented in \cite{GXY19}. Furthermore, the paper \cite{GMS21} provides three families of Singleton-optimal \LRCs based on matrix-product codes. The reader may refer to these papers for the detail.

 \subsection{Our results and comparison}
 Our paper studies maximum lengths of $q$-ary Singleton-optimal $[n,k,d;r]$-locally repairable codes for various parameter regimes with or without the constraint that  $(r+1)|n$ and recovery sets are disjoint.

 Our main results can be divided into two parts. The first part gives upper bounds on lengths of Singleton-optimal $[n,k,d;r]$-locally repairable codes without the constraint that  $(r+1)|n$ and recovery sets are disjoint, while the result of the second part requires this constraint.

 For two integers $m,\ell$ with $m\ge 2$, denote by $[\ell\pmod{m}]$ the principle remainder $s$ of $\ell$ divided by $m$ with $0\le s\le m-1$. Denote by $\{\ell\pmod{m}\}$ the  remainder $s$ of $\ell$ divided by $m$ with $1\le s\le m$.

 \begin{theorem}\label{thm:1.4}  If there is a $q$-ary Singleton-optimal $[n,k,d;r]$-locally repairable code, then
 \begin{itemize}
 \item[{\rm (i)}] (see Theorem \ref{thm:3.5})
 $n=O(q)$ if $d$ is proportional to $n$;
\item[{\rm (ii)}] (see Corollary \ref{cor:3.8})
$n\le q+O(1)$ if the Main MDS conjecture holds and $k$ is a constant and $k\not\equiv1\pmod{r}$ (note that in this case, $d$ must be proportional to $n$).
\item[{\rm (iii)}] (see Theorem \ref{thm:3.13}) If $r+2-\{k\pmod{r}\}<d<r+2$, we have $n< \frac{(d+r)(r-1+(d-2)r)}{r-d+2}$. In particular, $n=O(1)$ if both $r$ and $d$ are constant.
\item[{\rm (iv)}] (see Theorem \ref{thm:3.9} and Remark \ref{rmk:5}) If $dr=o(n)$,
 we have{\small
\[n\le \left\{\begin{array}{ll}
(1+o(1))\frac{r+1}{r}\times\frac{d-1}{4(q-1)}\times \min\left\{q^{4r(d-1-\Ge)/((d-t)(r+1))},q^{4(d-2)/(d-t)}\right\} &\mbox{if $d\equiv1, 2\pmod{4}$}\\
(1+o(1))\frac{r+1}{r}\times\frac{d-1}{4(q-1)}\times \min\left\{q^{4r(d-2-\Ge)/((d-t)(r+1))},q^{4(d-3)/(d-t)}\right\} &\mbox{if $d\equiv3, 4\pmod{4}$},
\end{array}
\right.\]}
where  $t=\{ d\pmod{4}\}$ and $\epsilon=\frac{\{k\pmod{r}\}}{r}$.
\item[{\rm (v)}] (see Theorems \ref{thm:3.10}-\ref{thm:3.12})  for $d=5,6,7$, we have\end{itemize}
%%%%%%%%%%%%%%%%%%%%%%%%%%%%%%%%%%%%%%%%%
{\footnotesize
\begin{center}
{\rm Table I\\  Upper bounds  for $d=5,6$ and $7$ without the constraint \medskip
\setlength{\tabcolsep}{0.7mm}
\begin{tabular}{|c|c|c|c|c|clclc|}\hline\hline
\multicolumn{8}{|c|}{$d=5$}\\ \hline
\multirow{2}{*}{$r$} &{\multirow{2}{*}{$1$}} & \multicolumn{2}{c|}{\multirow{2}{*}{$2,3$}} & \multicolumn{2}{c|}{$\ge 4$ and} & \multicolumn{2}{c|}{$\ge 4$ and}\\
 &&\multicolumn{2}{c|}{} & \multicolumn{2}{c|}{$k\equiv0,-1,-2\pmod{r}$} & \multicolumn{2}{c|}{$r=o(n)$}\\ \hline
 $n$&{$O(1)$}&\multicolumn{2}{c|}{$O(q)$}& \multicolumn{2}{c|}{$O(r)$}& \multicolumn{2}{c|}{$O(q^2)$}\\ \hline\hline
%%%%%%%%%%%%%%%%%%%%%%%%%%%
\multicolumn{8}{|c|}{$d=6$}\\ \hline
\multirow{2}{*}{$r$} & \multirow{2}{*}{$1$} &{{$2$ and}} & {{$2$ and}} & \multirow{2}{*}{$3,4$} & \multicolumn{2}{c|}{$\ge 5$ and}& \multicolumn{1}{c|}{$\ge 5$ and}\\
& & $k\equiv 0\pmod{2}$ & $k\equiv 1\pmod2$ & &  \multicolumn{2}{c|}{$k\equiv0,-1,-2,-3\pmod{r}$} & \multicolumn{1}{c|}{$r=o(n)$}\\ \hline
$n$&$O(q)$&$O(q)$&$O(q^2)$&$O(q^2)$&\multicolumn{2}{c|}{$O(r)$}&\multicolumn{1}{c|}{$O\left(q^{\frac{4r-2}{r+1}}\right)$}\\ \hline\hline
%$r$ & \begin{tabular}{@{}c|c@{}} $1$ & {\shortstack{$2$ and \\$k\equiv 0\pmod{2}$}} \\
              %   \end{tabular}  &  \begin{tabular}{@{}c|c@{}}  {\shortstack{$2$ and\\ $k\equiv 1\pmod2$}} & $3,4$  \end{tabular} & {\shortstack{$\ge 5$ and\\ $k\equiv0,-1,-2\pmod{r}$}} &{{$\ge 5 $ and $r=o(n)$}}\\ \hline
%$n$&$O(q)$&$O(q)$&$O(q)$&$O(q)$&$O(q^2)$&$O\left(q^{\frac{4r-2}{r+1}}\right)$\\ \hline\hline
%%%%%%%%%%%%%%%%%%%%%%%%%%%%%%
\multicolumn{8}{|c|}{$d= 7$}\\ \hline
\multirow{2}{*}{$r$} & \multirow{2}{*}{$1$} & \multirow{2}{*}{$2$} & {$3$}  &{$3$} & \multicolumn{1}{c|}{\multirow{2}{*}{$4,5$}}  &\multicolumn{1}{c|}{$\ge 6$ and} &\multicolumn{1}{c|}{$\ge 6$ and} \\
& & & $k\equiv 0\pmod{3}$ & \multicolumn{1}{c|}{$k\not\equiv 0\pmod{3}$} &\multicolumn{1}{c|}{}  & \multicolumn{1}{c|}{$k\equiv -1,-2,-3,-4\pmod{r}$}& \multicolumn{1}{c|}{$r=o(n)$} \\ \hline
$n$&$O(1)$&$O(q)$&$O(q)$&$O(q^2)$&\multicolumn{1}{c|}{$O(q^2)$}&\multicolumn{1}{c|}{$O(r)$} &\multicolumn{1}{c|}{$O\left(q^{\frac{4r-2}{r+1}}\right)$}\\ \hline\hline
\end{tabular}
}
\end{center}
}
 \end{theorem}

% Note that in the same way we can obtain maximum lengths of $q$-ary Singleton-optimal $[n,k,d=8;r]$-locally repairable codes.

 Theorem \ref{thm:1.4} (i) says that if minimum distance $d$ is proportional to length $n$, then $n$ is bounded by $O(q)$. So far,  in case $d$ is proportional to length $n$ the largest known code length  is $q+2\sqrt{q}+O(1)$ due to the construction via elliptic curves (see \cite{LMX20}). Thus, it is natural to ask the following open problem.
 \begin{open} What is the maximum length of a $q$-ary Singleton-optimal $[n,k,d;r]$-locally repairable code when the minimum distance $d$ is proportional to length $n$? More precisely speaking, if $d=\Gl n$ for some $\Gl\in (0,1)$, do we have $n\ge \frac{2q}{\Gl}$ ?
\end{open}
 Note that in Theorem \ref{thm:1.4} (ii), the length $n$ is upper bounded by $q+O(1)$ if $k$ is a constant and $k\not\equiv1\pmod{r}$. On the other hand,
 by using elliptic curves, Li et al. \cite{LMX20} provided a construction of  $q$-ary Singleton-optimal $[n,k,d;r]$-locally repairable codes with $n=q+2\sqrt{q}+O(1)$, $d=\Omega(n)$ and $k\equiv1\pmod{r}$. Thus,  the following open problem arises.
\begin{open} What is the maximum length of a $q$-ary Singleton-optimal $[n,k,d;r]$-locally repairable code if $k$ is a constant and $k\equiv1\pmod{r}$? Can the length exceed $n=q+2\sqrt{q}+O(1)$ in this case?
\end{open}

Based on Theorem \ref{thm:1.4} (iii), we checked  all known constructions of $q$-ary Singleton-optimal $[n,k,d;r]$-locally repairable codes with $n\rightarrow\infty$ appeared in literature and found that these constructions give code parameters beyond our range, i.e., $d\not\in (r+2-\{k\pmod{r}\},r+2)$. For instance, in \cite{Jin19}, $[n,k,d;r]$-locally repairable codes of length $\Omega(q^2)$ with $(d=5,r\ge 4, k\equiv -3\pmod{r})$ and $(d=6, r\ge 5, k\equiv -4\pmod{r})$ are constructed. Thus, $d=5\not\in (5,r+2)=(r+2-\{k\pmod{r}\},r+2)$ and $d=6\not\in (6,r+2)=(r+2-\{k\pmod{r}\},r+2)$, respectively. Another example is that $[n,k,d;r]$-locally repairable codes of length $\Omega(q)$ with $(d=6,r=2)$ are constructed in \cite{CFXHF}. For this example, we have $d\not\in (r+2-\{k\pmod{r}\},r+2)$.

Note that the upper bound  given in \cite{GXY19} is more restrictive compared with Theorem~\ref{thm:1.4} (iv). Our bound is less restrictive. Furthermore,  our bound improves the upper bound  given in \cite{GXY19} as well. This is because $q^{4r(d-1-\Ge)/((d-t)(r+1))}<q^{4(d-2)/(d-t)}$ and
$q^{4r(d-2-\Ge)/((d-t)(r+1))}<q^{4(d-3)/(d-t)}$ if $d>r+2$.

Note that in \cite{XY18} and \cite{CFXHF} they give some upper bounds for $d=5,6$ and $r=1,2,3$ and $r\ge 4$. Our Table I in Theorem~\ref{thm:1.4} (v) outperforms the upper bounds given in \cite{XY18} and \cite{CFXHF} for the some cases. The detailed comparison is given in the table below.
{\footnotesize
\begin{center}
Table II\\ Comparison of upper bounds  for $d=5$ with locality $r$ \\ \smallskip
{\rm
\begin{tabular}{|c|c|c|cl}\hline\hline
\multicolumn{4}{|c|}{Comparison for $d=5$}\\ \hline
${\rm Locality}$ &  \cite{XY18} & \cite{CFXHF} & \multicolumn{1}{c|}{Theorem \ref{thm:1.4} (v)} \\ \hline
$r=1$ &$O(q^2)$ & $O(q)$ & \multicolumn{1}{c|}{\bf ${\bf{\textit O(1)}}$} \\ \hline
{$r=2$} &$O(q^2)$ & {$O(q)$} &\multicolumn{1}{c|}{$O(q)$} \\ \hline
$r=3$ & $O(q^3)$ &$O(q^2)$ &\multicolumn{1}{c|}{\bf ${\bf{\textit O(\textit q)}}$} \\ \hline
{$r\ge 4$ and} & \multirow{2}{*}{$O(q^3)$} &  \multirow{2}{*}{$O(q^2)$} & \multicolumn{1}{c|}{\multirow{2}{*}{\bf ${\bf{\textit O(\textit r)}}$}} \\
{$k\equiv0,-1,-2\pmod{r}$} & & & \multicolumn{1}{c|}{}\\ \hline
$r\ge 4$ and $r=o(n)$ & $O(q^3)$ & $O(q^2)$ & \multicolumn{1}{c|}{$O(q^2)$} \\ \hline\hline
\end{tabular}
}
\end{center}
}
\bigskip\bigskip\bigskip\bigskip
%%%%%%%%%%%%%%%%%%%%%%%%%%%d=6
{\footnotesize
\begin{center}
Table III\\  Comparison of upper bounds  for $d=6$ with locality $r$ \\ \smallskip
{\rm
\begin{tabular}{|c|c|c|cl}\hline\hline
\multicolumn{4}{|c|}{Comparison for $d=6$}\\ \hline
${\rm Locality}$ & \cite{XY18} & \cite{CFXHF} &  \multicolumn{1}{c|}{Theorem \ref{thm:1.4} (v)} \\ \hline
$r=1$&$O(q^2)$ & $O(q)$ & \multicolumn{1}{c|} {\textbf{$O(q)$}} \\ \hline
$r=2$ and  &\multirow{2}{*}{$O(q^2)$} &  \multirow{2}{*}{$O(q^{1.5})$} &  \multicolumn{1}{c|}{ \multirow{2}{*}{\bf ${\bf{\textit O(\textit q)}}$}}\\
 $k\equiv 0\pmod{2}$ & &  & \multicolumn{1}{c|}{} \\  \hline
 $r=2$ and  &\multirow{2}{*}{$O(q^2)$} &  \multirow{2}{*}{\bf ${\bf{\textit O(\textit q^{1.5})}}$}&   \multicolumn{1}{c|}{\multirow{2}{*}{ $O(q^2)$}}\\
 $k\equiv 1\pmod{2}$ & &  & \multicolumn{1}{c|}{} \\  \hline
$r=3$ &$O(q^2)$& $O(q^2)$ &  \multicolumn{1}{c|}{$O(q^2)$} \\ \hline
$r= 4$ &$O(q^3)$&  $O(q^3)$ &  \multicolumn{1}{c|}{\bf ${\bf{\textit O(\textit q^2)}}$}\\ \hline
 $r\ge 5$ and   &\multirow{2}{*}{$O(q^3)$}&  \multirow{2}{*}{$O(q^{3})$} &   \multicolumn{1}{c|}{\multirow{2}{*}{\bf ${\bf{\textit O(\textit r)}}$}}\\
 $k\equiv 0,-1,-2,-3\pmod{r}$  & &  & \multicolumn{1}{c|}{} \\  \hline\hline
\end{tabular}
}
\end{center}
}

{It seems that Propositions \ref{prop:1.1} and \ref{prop:1.2} give better upper bound than Table I for the cases where $(d,r)=(6,\ge 5)$ and $(d,r)=(7,\ge 6)$. However, both  Propositions \ref{prop:1.1} and \ref{prop:1.2} require the condition that $(r+1)|n$ and recovery sets are disjoint. Thus, the bounds given by \ref{prop:1.1} and \ref{prop:1.2} are not compatible with the upper bound given in Table I.}

 {\begin{open} What is the maximum length of a $q$-ary Singleton-optimal $[n,k,d;r]$-locally repairable code for the case $(d,r)=(6,\ge 5)$ and $(d,r)=(7,\ge 6)$ if we do not assume the condition that $(r+1)|n$ and recovery sets are disjoint.
\end{open}}

Recall that both the papers \cite{Jin19} and \cite{GXY19} presented constructions of  $q$-ary Singleton-optimal $[n,k,d;r]$-locally repairable codes with $d=5$ and $6$. {For instance, in \cite{Jin19}, two classes of codes, namely  $q$-ary Singleton-optimal $[n,k,d;r]$-locally repairable codes with (i) $d=5$, $4\le r\le q-1,$ $(r+1)|n$ and $k\equiv-3\pmod{r}$; and (ii)  $d=6$, $r\ge 5,$ $(r+1)|n$ and $k\equiv-4\pmod{r}$, are constructed. Furthermore, both have length upper bounded by $O(q^2)$.} On the other hand, from  Table I, we know that a $q$-ary Singleton-optimal $[n,k,d;r]$-locally repairable codes with (i) $d=5$, $r\ge 4$ and $k\not\equiv-3\pmod{r}$; {and (ii)  $d=6$, $r\ge 5$ and $k\not\equiv-4\pmod{r}$ has length upper bounded by $O(r)$.}

We note that Theorem \ref{thm:1.4} does not assume that $(r+1)|n$ and recovery sets are disjoint. However, our second result requires this constraint.
By a simple propagation rule, we are able to derive upper bounds on  maximum length of a $q$-ary Singleton-optimal $[n,k,d;r]$-locally repairable codes for relatively large minimum distance and locality.

\begin{theorem}\label{thm:1.5} If there is a $q$-ary Singleton-optimal $[n,k,d;r]$-locally repairable code satisfying that $n\ge\Omega(dr^2)$, $(r+1)|n$ and  recovery sets are disjoint, then
\begin{itemize}
%\item[{\rm (i)}]  (see Theorem \ref{thm:4.8}) If $ d\equiv4t-i \pmod{ r+1}$
%with $t\geq 1$ and $i\in \{-1,0,1,2\}$, we have $n=O(q^{3-\frac{1}{t}})$.
%\item[{\rm (ii)}]  (see Theorem \ref{thm:4.7}) If $d\equiv1,2,3,4,5 \pmod {r+1}$, we have $n=O(q^{2})$.
\item[{\rm (i)}] (see Theorem \ref{thm:4.5}) for $r=1,2,3,4$, we have upper bounds in the following table
\end{itemize}
\bigskip\bigskip\bigskip\bigskip\bigskip\bigskip
\begin{center}
{\rm Table IV\\
Upper bounds  for $r=1,2,3,4$\\ \smallskip
\setlength{\tabcolsep}{0.5mm}
\begin{tabular}
%{|c|c|c|c|c|c|}
{|m{1cm}<{\centering}|m{2.7cm}<{\centering}|m{2.7cm}<{\centering}|m{2.7cm}<{\centering}|m{3.15cm}<{\centering}|m{2.7cm}<{\centering}|}
\hline\hline $r$ & $1$ & $1$ & $2$ & $2$ & $2$\\\hline $d$ & $d\equiv0\pmod{2}$ & $d\equiv1\pmod{2}$ & $d\not\equiv0\pmod{3}$ & {\shortstack{$d\equiv0\pmod{3}$\\ with $2|k$}} & $d\equiv0\pmod{3}$ \\\hline  $n$ & $O(q)$ & $O(1)$ & $O(q)$ &$O(q)$&$O(q^{1.5})$ \\\hline\hline $r$ & $3$ & $3$ & $3$ & $3$ & $3$ \\\hline $d$ & $d\equiv0\pmod{4}$ & {\shortstack{$d\equiv0\pmod{4}$\\ $k\not\equiv1\pmod{3}$}} & $d\equiv1\pmod{4}$  & $d\equiv2,3\pmod{4}$&{\shortstack{$d\equiv3\pmod{4}$\\ with $3|k$}}\\\hline $n$ & $O(q^2)$ & $O(q^2)$ &$O(q)$&$O(q^{2})$&$O(q)$\\\hline\hline $r$ & $4$ & $4$ & $4$ & $4$ & $4$ \\\hline $d$ & $d\equiv0\pmod{5}$ & $d\equiv1\pmod{5}$ & $d\equiv2\pmod{5}$  & $d\equiv3\pmod{5}$&$d\equiv4\pmod{5}$ \\\hline $n$ & $O(q^2)$ & $O(q^2)$ &$O(q^2)$&$O(q^{3})$&$O(q^3)$
\\ \hline\hline\end{tabular}
}\end{center}
\begin{itemize}
\item[{\rm (ii)}]  (see Theorem \ref{thm:4.7}) if $d\equiv1,2,3,4,5 \pmod {r+1}$, we have $n=O(q^{2})$.
\item[{\rm (iii)}]  (see Theorem \ref{thm:4.8}) if $ d\equiv4t-i \pmod{ r+1}$ with $t\geq 1$ and $i\in \{-1,0,1,2\}$, we have $n=O(q^{3-\frac{1}{t}})$.
\end{itemize}

 \end{theorem}

 \subsection{Organization of the paper}
The paper is organized as follows. In Section 2, we introduce some background on classical block codes as well as locally repairable codes. In Section 3, we derive a few upper bounds on  maximum length of  $q$-ary Singleton-optimal $[n,k,d;r]$-locally repairable codes without the constraint that $(r+1)|n$ and recovery sets are disjoint. In the last section, we produce some upper bounds assuming that $(r+1)|n$ and recovery sets are disjoint.

\section{Preliminaries}\label{sec:2}
\subsection{Some basic notations and results for classical block codes}
In this subsection, we introduce some definitions and results of classical block codes as follows.
Let $q$ be a power of an arbitrary prime and let $\F_q$ be the finite field with $q$ elements.

\begin{defn}
The support of a vector ${\bf u}=(u_1,\cdots, u_n)\in\F_q^n$ is defined by
$${\rm{supp}}({\bf u})=\{i\in[n]: u_i\neq 0\},$$
where $[n]=\{1,2,\cdots, n\}.$ The Hamming weight $\rm{wt}(\bf u)$ of ${\bf u}$ is defined to be the size of $\rm{supp}({\bf u}).$
\end{defn}

\iffalse
{\color{red}\begin{defn}
A $q$-ary code of length $n$ is called a locally repairable code (LRC for short) with locality $r$ if for any $i\in[n],$ there exists a subset $R\subset [n]\setminus\{i\}$ of size $r$ such that for any ${\bf c}=(c_1,\cdots, c_n)\in\mC,$ $c_i$ can be recovered by $\{c_j\}_{j\in R},$ i.e., for any $i\in[n],$ there exists a subset $R\subset [n]\setminus\{i\}$ of size $r$ such that for any ${\bf u, v}\in\mC, {\bf u}_{R\cup \{i\}}={\bf v}_{R\cup \{i\}}$ if and only if ${\bf u}_R={\bf v}_R.$ The set $R$ is called a recover set of $i.$
\end{defn}}\fi

A linear code $\mC$ of length $n,$ dimension $k$ and minimum distance $d$ over $\F_q$ is called a $q$-ary $[n,k,d]$-linear code.
Let $A_q(n,d)$ denote the largest possible size $M$ for which there exists an $(n,M,d)$-code over $\F_q.$  Then, we recall two important bounds: Hamming  and Griesmer bounds (see in~\cite{LX04}).

\begin{lemma}[Hamming bound]\label{lem:2.1}
For an integer $q>1$ and integers $n, d$ such that $1\le d\le n,$ we have
\begin{eqnarray}\label{eq:3}
A_q(n,d)\le\frac{q^n}{\sum_{i=0}^{\lfloor(d-1)/{2}\rfloor}{n\choose i}(q-1)^i}.
\end{eqnarray}
\end{lemma}

\begin{lemma}[Griesmer bound]\label{lem:2.2}
Let $\mC$ be an $[n,k,d]$-linear code over $\F_q,$ where $k\ge 1,$ then
\begin{eqnarray}\label{eq:4}
n\ge\sum_{i=0}^{k-1}\left\lceil{\frac{d}{q^i}}\right\rceil.
\end{eqnarray}
\end{lemma}

Now, we introduce some propagation rules~\cite{LX04} of classical block codes that will be used in this paper.
\begin{lemma}\label{lem:2.3}
Suppose there is an $[n,k,d]$-linear code over $\F_q.$ Then
\begin{itemize}
\item[($i$)] there exists an $[n-s, k, d-s]$-linear code over $\F_q$ for any $1\le s\le d-1;$
\item[($ii$)] there exists an $[n-s, k-s, d]$-linear code over $\F_q$ for any $1\le s\le k-1.$
\end{itemize}
\end{lemma}

Then, we recall the Singleton defect of an $[n,k,d]$-linear code $\mC$.
\begin{defn}
The Singleton defect of an $[n,k,d]$-linear code $\mC$ is defined to be $s(\mC)=n-k+1-d.$
\end{defn}
\begin{itemize}
\item[(i)]
A code $\mC$ with $s(\mC)=0$ is called {\it maximum distance separable} (MDS). Note that MDS codes with dimension $k\in\{0,1, n-1, n\}$ are called {\it trivial}. Then we have the following well-know conjecture.\\
\fbox{\begin{minipage}{34em}
%\begin{quote}
{\bf Main  MDS Conjecture ~\cite{MS}:} For a nontrivial $[n, k]$-MDS code we have that $n\le q+2$ if $q$ is even and $k=3$ or $k=q-1,$ and $n\le q+1$ otherwise. So, the length $n$ of a MDS code is always upper bounded by $q+2.$
%\end{quote}
\end{minipage}
}\\
\item[(ii) ] A code $\mC$ with $s(\mC)=1$ is called {\it an almost MDS} (AMDS) code.
Let $\mu(r,q)$ denote the maximum length $n$ for which there exists an $[n, n-r-1, r+1]$ code over $\F_q$, then we have
$\mu(2,q)=q^2+q+1 $ \cite{BB52}. Furthermore, since $\mu(r,q)\le \mu(r-1,q)+1$, we have $\mu(3,q)\le q^2+q+2$~\cite{BB52}.
\end{itemize}

\begin{rmk} If $\mC$ is a $q$-ary nontrivial $[n,k,d\ge n-k]$-linear code, then $d=n-k$ or $n-k+1$. If $d=n-k$, then $n\le \mu(n-k-1,q)$. If $d=n-k+1$, we can get an $[n,k,d=n-k]$-almost MDS code by setting the last coordinate to be $0$. Hence, we also have $n\le \mu(n-k-1,q)$.

\end{rmk}
\subsection{Some basic notations and results for locally repairable  codes}
 Informally speaking, if  every coordinate of a given codeword of a block code with locality $r$ can be recovered by accessing at most $r$ other coordinates of this codeword. The formal definition of a locally repairable code with locality $r$ is given as follows.

\begin{defn}{\rm
Let ${\mC}\subseteq \F_q^n$ be a $q$-ary block code of length $n$. For each $\Ga\in\F_q$ and $i\in \{1,2,\cdots, n\}$, define ${\mC}(i,\Ga):=\{\bc=(c_1,\dots,c_n)\in {\mC},\; c_i=\Ga\}$. For a subset $I\subseteq \{1,2,\cdots, n\}\setminus \{i\}$, we denote by ${\mC}_{I}(i,\Ga)$ the projection of ${\mC}(i,\Ga)$ on $I$.
 Then ${\mC}$
 is called a locally repairable code with locality $r$ if, for every $i\in \{1,2,\cdots, n\}$, there exists a subset
$I_i\subseteq \{1,2,\cdots, n\}\setminus \{i\}$ with $|I_i|\le r$ such that  ${\mC}_{I_i}(i,\Ga)$ and ${\mC}_{I_i}(i,\Gb)$ are disjoint for any $\Ga\neq \Gb$.}
\end{defn}
One can have an equivalence definition as the following.
\begin{defn}{\rm A $q$-ary code of length $n$ is called a locally recoverable code  with locality $r$ if for any $i\in[n]$,  there exists a subset $R\subset[n]\setminus\{i\}$ of size $r$ such that  for any  $\bu,\bv\in {\mC}$, $\bu_{R\cup\{i\}}=\bv_{R\cup\{i\}}$ iff $\bu_R=\bv_R$. For the sake of convenience, $R\cup\{i\}$ is called a recovery set of $i$ in the paper.
}\end{defn}

Thus, apart from the usual parameters: length, rate and minimum distance,  the locality of a  locally repairable code plays a crucial role. In this paper, we always consider locally repairable codes that are linear over $\F_q$. Thus, a $q$-ary \LRC of length $n$, dimension $k$, minimum distance $d$ and locality $r$ is said to be an $[n,k,d]_q$-locally repairable code with locality $r$, denoted by $[n,k,d;r]$ or $[n,k,d;r]_q$ if we want to emphasize the code alphabet  size $q$.

 If we ignore minimum distance of   a $q$-ary locally repairable code, then there is a constraint on the rate \cite{GHSY12}, namely,
\begin{equation*}
\frac{k}{n} \le \frac{r}{r+1}.
\end{equation*}

However, if we also care about global error correction, we have to take minimum distance into account. In this case, the parameters have the constraint given in \eqref{eq:1}.

The following result can be found in \cite{GXY19} .

\begin{lemma}\label{lem:2.4} Let $n,k,d,r$ be positive integers with $(r+1)|n$. If the Singleton-type bound \eqref{eq:1} is achieved, then
\begin{equation*}
n-k=\frac{n}{r+1}+d-2-\left\lfloor\frac{d-2}{r+1}\right\rfloor.
\end{equation*}
\end{lemma}

\begin{lemma}\label{lem:2.5} Given a Singleton-optimal  $[n,k,d;r]$-locally repairable code, the dimension $k$ is completely determined by $n,d,r$.
\end{lemma}
\begin{proof} Suppose we have two Singleton-optimal  locally repairable codes $[n,k_1,d;r]$ and $[n,k_2,d;r]$ with $1\le k_1<k_2$. Then we have
\[d=n-k_1-\left\lceil\frac {k_1}r\right\rceil+2>n-k_2-\left\lceil\frac {k_2}r\right\rceil+2=d.\]
This is a contradiction.
\end{proof}

For linear codes, the following lemma establishes a connection between the locality and the dual code ${\mC}^{\perp}$. It was proved in \cite{GXY19} that  recovery sets are determined by codewords with specific support in the dual code. Precisely we have Lemma~\ref{lem:2.6}.

\begin{lemma}[see  \cite{GXY19}]\label{lem:2.6} A subset $R\subset[n]$ with $i\in R$ is a recovery set at $i$ for a $q$-ary linear code of length $n$ if and only if there exists a codeword $\bc\in {\mC}^{\perp}$ such that $i\in\supp(\bc)\subset R$.
\end{lemma}

For a linear code ${\mC}$  and $r\ge 1$,  we define  the set
\[\mathfrak{R}_{\mC}(r):=\{\supp(\bc):\; \bc\in {\mC}^\perp,\; |\supp(\bc)|\le r+1\}.\]
The following result is a straightforward corollary of Lemma \ref{lem:2.2}.

\begin{cor}\label{cor:2.7} A linear code ${\mC}$ has locality $r$ if and only if $[n]=\cup_{I\in \mathfrak{R}_{\mC}(r)}I$.
\end{cor}

Finally, we introduce the disjoint recovery sets.
\begin{defn}{\rm { Let $\ell$ be an integer.} We say that a linear code ${\mC}$ with locality  $r$ has  disjoint recovery sets if there exist subsets $\{I_1, I_2, \cdots, I_\ell\}\subseteq \mathfrak{R}_{\mC}(r)$ that form a partition of $[n]$.
}\end{defn}

\section{Upper Bounds for Locally Repairable Codes}~\label{sec:3}
In this section, we focus on upper bounds for code length when $d=\Omega(n)$ and  $d=o(n)$, respectively. We also provide an upper bound for general minimum distance $d$ and locality $r$.

Throughout the paper, we assume that parameters of a Singleton-optimal $[n,k,d;r]$-locally repairable codes  satisfy
 \begin{equation}\label{eq:7}
  r<k ,\quad {n\ge 2(r+1)}\quad {\rm and}\quad d\ge 3.
 \end{equation}
To start, we construct a parity-check matrix $H$ of an $[n,k,d; r]$-linear locally repairable code  $\mC$ with a specific form.
\\

\fbox{\begin{minipage}{35em}
Based on the Lemma~\ref{lem:2.6}, we can choose a basis of ${\mC}^{\perp}$ as follows: there exists a vector $\bu_1\in {\mC}^{\perp}$ such that $i_1=1\in\supp(\bu_1)$ and $|\supp(\bu_1)|\le r+1$. Choose $i_2\in[n]\setminus\supp(\bu_1)$, then  there exists a vector $\bu_2\in {\mC}^{\perp}$ such that $i_2\in\supp(\bu_2)$ and $|\supp(\bu_2)|\le r+1$. Choose $i_3\in[n]\setminus(\supp(\bu_1)\cup(\bu_2))$, then  there exists a vector $\bu_3\in {\mC}^{\perp}$ such that $i_3\in\supp(\bu_3)$ and $|\supp(\bu_3)|\le r+1$.
Continue in this fashion until we get $\bu_1,\bu_2,\dots,\bu_\ell\in {\mC}^\perp$ satisfying:
\begin{itemize}
\item[(i)] $|\supp(\bu_i)|\le r+1$ for all $i\in[\ell]$;
\item[(ii)] $\bigcup_{i=1}^\ell\supp(\bu_i)=[n]$;
\item[(iii)] $\supp(\bu_i)\nsubseteq \cup_{j=1}^{i-1}\supp(\bu_j)$ for all $i\in\{2,3,\dots,\ell\}$.
\end{itemize}
From the above condition (iii), it is clear that  $\bu_1,\bu_2,\dots,\bu_\ell$ are linearly independent. Thus, $\ell\le n-k$ and we can extend to a basis $\bu_1,\bu_2,\dots,\bu_{n-k}$ of ${\mC}^{\perp}$. Denote by $H_1\in\F_q^{\ell\times n}$ and $H_2\in\F_q^{(n-k-\ell)\times n}$, respectively the matrix whose rows consists of  $\bu_1,\bu_2,\dots,\bu_\ell$ and $\bu_{\ell+1},\bu_{\ell+2},\dots,\bu_{n-k}$, respectively. Denote by $H$ the matrix
\begin{equation}\label{eq:8}
H=\begin{pmatrix}H_1\\ H_2\end{pmatrix}.
\end{equation}
Note that the number of rows of $H_1$ is $\ell$. We denote by $h$ the number of rows of $H_2$. Then
\begin{equation}\label{eq:9} h+\ell=n-k.\end{equation}
\end{minipage}
}\\
\medskip

Based on the above construction of a parity-check matrix $H$, we have the following inequalities for parameters involved.
\begin{lemma}~\label{lemma:3.1} Let
$H\in\F_q^{(n-k)\times n}$ be a parity-check matrix of $\mC$ given in \eqref{eq:8}. We have
\begin{equation}\label{eq:10}
\frac kr< \frac n{r+1}\le \ell\le n-k.
\end{equation}
\begin{proof}
We write $H=(h_{ij})\in\F_q^{(n-k)\times n}, 1\leq i\leq n-k, 1\leq j\leq n$. Since $\bu_1,\bu_2,\dots,\bu_\ell$ are linearly independent, so it is clear $\ell\le n-k$. Consider the first $\ell$ rows of $H,$ each rows contains at most $r+1$ nonzero elements, so the total number of nonzero elements within the first $\ell$ rows of $H$ is at most $\ell(r+1)$. On the other hand, for every $j\in[n]$ there is $i\in[\ell]$ such that the element $h_{ij}\neq 0.$ This means that within the first $\ell$ rows, there are at least $n$ nonzero elements, so $n\le \ell(r+1)$. Therefore, we have
\begin{equation}\label{eq:11}
\nonumber\ell+k\le n\le \ell(r+1).
\end{equation}
 By the assumption  $d\ge 3$ and  the Singleton-type bound, we have
\begin{equation}
\nonumber3\le d\le n-k-\left\lceil\frac kr\right\rceil+2\le n-k-\frac kr+2,
\end{equation}
i.e., $ \frac kr< \frac n{r+1}$.
Hence, we have
\begin{equation*}
\frac kr< \frac n{r+1}\le \ell\le n-k.
\end{equation*}
\end{proof}
\end{lemma}
%%%%%%%%%%%%%%%
We shall construct a parity-check matrix $H_I$. Let $\mC$ be an $[n,k,d; r]$-locally repairable code over $\F_q$ with a parity-check matrix $H=(h_{ij})\in\F_q^{(n-k)\times n}$ given in \eqref{eq:8}.
\\

\fbox{\begin{minipage}{35em}
 Let $I$ be a subset of $[\ell],$ define the set $J_{I}=\{j\in[n]:h_{ij}\neq 0 {~\rm ~for ~some}~ i\in I\}.$ Denote by $H_I$ the submatrix obtained from $H$ by removing all the rows of $H$ indexed by $I$ and all the columns of $H$ indexed by $J_I$. \end{minipage}
}\\
\medskip
\begin{lemma}~\label{lemma: C_I}
Let $\mC_I$ be the linear code with $H_I$ as a parity-check matrix, then $\mC_I$ is a $q$-ary $[n_I\ge n-|I|(r+1), k_I\ge k-r|I|, d_I\ge d]$-linear code.\end{lemma}

\iffalse
\begin{lemma}\label{lem:3.2}
Let $\mC$ be an $[n,k,d; r]$-locally repairable code over $\F_q$ with a parity-check matrix $H=(h_{ij})\in\F_q^{(n-k)\times n}$ given in \eqref{eq:8}. Let $I$ be a subset of $[\ell],$ define the set $J_{I}=\{j\in[n]:h_{ij}\neq 0 {~\rm ~for ~some}~ i\in I\}.$ Denote by $H_I$ the submatrix obtained from $H$ by removing all the rows of $H$ indexed by $I$ and all the columns of $H$ indexed by $J_I$. Let $\mC_I$ be the linear code with $H_I$ as a parity-check matrix, then $\mC_I$ is a $q$-ary $[n_I, k_I, d_I]$-linear code with
$$n_I\ge n-|I|(r+1);~k_I=n_I-(n-k-|I|)\ge k-r|I|~{\rm and}~d_I\ge d.$$
\end{lemma}
\fi
\begin{proof} As each row indexed by $i\in I$ has at most $r+1$ nonzero elements, thus we delete at most $|I|(r+1)$ columns and hence the length $n_I$ is at least $n-|I|(r+1)$. The dimension $k_I$ of ${\mC}_I$ is $n_I-{\rm rank}(H_I)\ge n_I-(n-k-|I|)\ge k-r|I|$.
\iffalse
Now, we prove the minimum distance of $\mC_I$ should be at least $d.$ Note that $H_I$ can be obtained in two steps: (i) deleting the columns of $H$ indexed by $J_I$ to get a matrix $H^*$. Then it is clear that any $d-1$ columns of $H^*$ are linearly independent. {Now all the entries $h_{ij}$ with $i\in I$ are equal to zero.} (ii) deleting rows of $H^*$ gives $H_I$. Since we delete all zero entries of $H^*$, linear independence is unchanged, i.e.,
 any $d-1$ columns of $H_I$ are still linearly independent. Thus, we get the desired lower bound on minimum distance.
 \fi

 Now, we prove the minimum distance of $\mC_I$ should be at least $d.$ Note that $H_I$ can be obtained in two steps: {(i) deleting the columns of $H$ indexed by $J_I$ to get the remaining matrix $H^*$. Then it is clear that any $d-1$ columns of $H^*$ are linearly independent. (ii) deleting zero rows of $H^*$ gives $H_I$. }Since we delete all zero entries of $H$, so the linear independence is unchanged, i.e.,
 any $d-1$ columns of $H_I$ are still linearly independent. Thus, we get the desired lower bound on minimum distance.
\end{proof}

\subsection{Upper bound on code length when $d=\Omega(n)$} Throughout this subsection, we assume that parameters of a Singleton-optimal $[n,k,d;r]$-locally repairable code  satisfies $r\ge 2$.
Recall that an $[n,k,d]$-linear code is called trivial if $k=0,1,n-1$ or $n$.

\begin{lemma}\label{lem:3.3}
Assume that ${\mC}$ is {a Singleton-optimal} $[n,k,d;r]$-locally repairable code. Let $I$ be a subset of $[\ell]$ with $|I|=\left\lceil\frac kr\right\rceil-2$ (note that such an $I$ can be chosen as $\ell\ge \left\lceil\frac kr\right\rceil> \left\lceil\frac kr\right\rceil-2\ge 0$).
 If  $d\ge 3$, then ${\mC}_I$ is a nontrivial $[n_I,{k_I=n_I-d}, d_I\ge d]$-linear code.
\end{lemma}
\begin{proof}  The number of rows of $H_I$ is $n-k-|I|=n-k-\left(\left\lceil\frac kr\right\rceil-2\right)=d$ and any $d-1$ rows of $H_I$ are linearly independent. We conclude that  ${\mC}_I$ is an $[n_I,k_I= n_I-d, d_I\ge d]$-linear code.

Next we show that ${\mC}_I$ is nontrivial. Based on Lemma~\ref{lemma: C_I}, we have $k_I\ge k-r|I|=k-r\left(\left\lceil\frac kr\right\rceil-2\right)\ge k-r\left(\frac kr+1-2\right)=r\ge 2$. Furthermore, we have
\[k_I=n_I-{\rm rank}(H_I)\le n_I-(d-1)\le n_I-2.\] This  completes the proof.
\end{proof}

\begin{lemma}\label{lem:3.4}~\cite{MS} If there is a $q$-ary $[n, {k= n-d}, \ge d]$-linear code $\mC$ with $k\ge 2$, then $d\le 2q$.
\end{lemma}
\begin{proof} Suppose $d>2q$. Applying the Grismer bound in~(\ref{eq:4}) to the code gives
\[n\ge \sum_{i=0}^{k-1}\left\lceil\frac d{q^i}\right\rceil\ge d+3+(k-2)=d+k+1\ge n+1.\]
This is a contradiction.
\end{proof}

\begin{theorem}\label{thm:3.5} If there exists {a Singleton-optimal} $[n,k,d;r]$-locally repairable code with $d=\Omega(n)$, then $n=O(q)$. More precisely, if $d= \Gl n$ for some $\Gl<1$, then $n\le \frac{2q}{\Gl}$.
\end{theorem}
\begin{proof} As $d=\Omega(n)$, there exists a real $\Gl>0$ such that $d= \Gl n$. Choose a subset $I$ of $[\ell]$ with $|I|=\left\lceil\frac kr\right\rceil-2$. Applying Lemma \ref{lem:3.4} to the code ${\mC}_I$ given in Lemma \ref{lem:3.3}, we get $d\le 2q$. Thus, we have
\[n\le \frac{d}{\Gl}\le \frac{2}{\Gl}\times q.\]
The proof is completed.
\end{proof}

{Next we give a sufficient condition under which $n\le q+O(1)$ in Theorem~\ref{thm:3.7}.}

\begin{lemma}\label{lem:3.6}
 If $\mC$ is {a Singleton-optimal} $[n,k,d;r]$-locally repairable code, then ${\mC}_I$ is an MDS code if $|I|=\left\lceil\frac kr\right\rceil-1$. Furthermore,  if $k\not\equiv1\pmod{r}$ and $d\ge 3$, then ${\mC}_I$ is a nontrivial MDS code.
\end{lemma}
\begin{proof} Firstly, we show that ${\mC}_I$ is an MDS code if $|I|=\left\lceil\frac kr\right\rceil-1$. Note that the parity-check matrix $H_I$ of ${\mC}_I$ has $n-k-|I|=n-k-\left\lceil\frac kr\right\rceil+1=d-1$ rows. Since any $d-1$ columns of $H_I$ is linearly independent, so ${\mC}_I$ is an MDS code.

  Assume that $k\equiv s\pmod{r}$ with $0\le s\le r-1$. From Lemma~\ref{lemma: C_I},  we have $k_I\ge k-r|I|.$ If $s=0$, i.e., $r|k$, then $k_I\ge k-r|I|=k-(k-r)=r>1$. If $s\neq 0$, then $s\ge 2$ and  we have $k_I\ge k-r|I|=k-(k+r-s-r)=s>1$. Furthermore, as $k_I= n_I-(n-k-|I|)$, we have $n_I-k_I=n-k-|I|=d-1\ge 2$. This completes the proof.
\end{proof}

\begin{theorem}\label{thm:3.7} Suppose that the Main MDS conjecture holds. Assume that there exists {a Singleton-optimal} $[n,k,d;r]$-locally repairable code $\mC$, if $k\not\equiv1\pmod{r}$ then the length of $\mC$ is bounded by
 \begin{equation*}n\le q+1+k+\frac{k}{r}.\end{equation*}
\end{theorem}
\begin{proof} Let ${\mC}$ be {a Singleton-optimal} $[n,k,d;r]$-locally repairabel code. Choose a subset $I$ of $[\ell]$ with $|I|=\left\lceil\frac kr\right\rceil-1$.
By the Main MDS conjecture ~\cite{MS}, and Lemmas~\ref{lemma: C_I} and \ref{lem:3.6}, we have
\[q+2\ge n_I\ge n-|I|(r+1)>n-\frac{k}{r}(r+1)=n-k-\frac{k}{r}.\]
This gives $n\le q+1+k+\frac{k}{r}$.
\end{proof}

By Theorem~\ref{thm:3.7}, we have following results.
\begin{cor}\label{cor:3.8} Suppose that the Main MDS conjecture holds. If there exists {a Singleton-optimal} $[n,k,d;r]$-locally repairable code with $k\not\equiv1\pmod{r}$, then we have
\begin{itemize}
\item[{\rm (i)}] if $k$ is a constant, then $n\le q+O(1)$.
\item[{\rm (ii)}] if $d=\Gl n$ for some $\Gl<1$, then $n\le \frac{q}{\Gl}+O(1).$
\end{itemize}
\end{cor}
\begin{proof} If $k$ is a constant, then  by Theorem \ref{thm:3.7}, we have $n\le q+1+k+\frac{k}{r}=q+O(1)$.
On the other hand,
if $d=\Gl n$, then by the Singleton-type bound in~(\ref{eq:1}), we have
\[d=\Gl n=n-k-\left\lceil\frac kr\right\rceil+2=n-k-\frac kr+O(1),\]
i.e., $k=\frac{(1-\Gl)r}{r+1}n+O(1)$. By  Theorem \ref{thm:3.7}, we have
\[n\le q+1+\left(1+\frac1r\right)k=q+\left(1+\frac1r\right)\times\frac{(1-\Gl)r}{r+1}n+O(1)=q+(1-\Gl)n+O(1).\]
This gives the desired result.
\end{proof}

{\begin{rmk}
From \cite{LMX20}, there exists {a Singleton-optimal} $[n=(r+1)m,k=1+r(t-1),d=n-(t-1)(r+1);r]$-locally repairable code for $r=2,3,5,7,11,23$ and $1\le t<m=\lfloor\frac{q+2\sqrt{q}-r-2}{r+1}\rfloor$. Thus, $n=q+2\sqrt{q}+O(1)$ when $k\equiv1\pmod{r}$. Hence, to have the upper bound $n\le {q}+O(1)$, the condition $k\not\equiv1\pmod{r}$ given in part (i) of Corollary \ref{cor:3.8} is indeed  necessary.
\end{rmk}}

\begin{rmk}
 Compared with Theorem \ref{thm:3.5}, part (ii) of Corollary \ref{cor:3.8} gives a better upper bound $n\le {q}+O(1)$ if $d=\Omega(n)$.
\end{rmk}

\iffalse
\begin{rmk} \begin{itemize}
\item[(i)] The condition $k\not\equiv1\pmod{r}$ given in Theorem \ref{thm:3.7} is indeed a necessary condition.
From \cite{LMX20}, there exists {a Singleton-optimal} $[n=(r+1)m,k=1+r(t-1),d=n-(t-1)(r+1);r]$-locally repairable code for $r=2,3,5,7,11,23$ and $1\le t<m=\lfloor\frac{q+2\sqrt{q}-r-2}{r+1}\rfloor$. Thus, we have $k\equiv1\pmod{r}$, but $n=q+2\sqrt{q}+O(1)$.
%\item[(ii)] Compared with Theorem \ref{thm:3.5}, part (ii) of Corollary \ref{cor:3.8} gives a better upper bound $n\le \frac{q}{\Gl}+O(1),$ for some $\Gl<1$, if $d=\Omega(n)$.
\end{itemize}
\end{rmk}
\fi

\subsection{Upper bound on code length when $d=o(n)$}
In \cite{GXY19}, the authors derived an upper bound on maximum length of Singleton-optimal locally repairable codes by assuming that $(r+1)|n$ and recovery sets are disjoint. In this subsection, we generalize the upper bound given in \cite{GXY19} by removing the condition that $(r+1)|n$ and recovery sets are disjoint.

%The first result that we are presenting in this subsection is to generalize the upper bound given in \cite{GXY19} by removing the condition that $(r+1)|n$ and recovery sets are disjoint.
\begin{theorem}~\label{thm:3.9} If there exists {a Singleton-optimal} $[n,k,d;r]$-locally repairable code with $rd=o(n)$, then we have
\[n\le \left\{\begin{array}{ll}
(1+o(1))\frac{r+1}{r}\times\frac{d-1}{4(q-1)}q^{4r(d-1-\Ge)/((d-t)(r+1))} &\mbox{if $d\equiv1, 2\pmod{4}$}\\
(1+o(1))\frac{r+1}{r}\times\frac{d-1}{4(q-1)}q^{4r(d-2-\Ge)/((d-t)(r+1))} &\mbox{if $d\equiv3, 4\pmod{4}$},
\end{array}
\right.\]
where $t=\{ d\pmod{4}\}$ and $\epsilon=\frac{\{k\pmod{r}\}}{r}$.
\end{theorem}
\begin{proof}
Assume that $\mC$ is {a Singleton-optimal} $[n,k,d;r]$-locally repairable code {attaining the Singleton bound} with $d=o\left(\frac{n}{r}\right)$. Let $k\equiv -s\pmod{r}$ with $0\le s\le r-1$, i.e., $s=r-\{k\pmod{r}\}$. Then we have
\[d=n-k-\left\lceil\frac kr\right\rceil+2=n-k-\frac{k+s}{r}+2.\]
This gives
\begin{equation}\label{eq:12}\frac{r+1}r\times \frac kn=1-\frac dn+\frac1n\left(2-\frac sr\right).\end{equation}
Put $f(n)=\frac dn-\frac1n\left(2-\frac sr\right)$. As $d\ge 3,$ we have $f(n)>0$.
%By \eqref{eq:12}, we have $f(n)>0$. % and $f(n)\rightarrow0$ as $n\rightarrow\infty$.
Rewrite $f(n)$ into the following identity
%{\color{blue}\begin{equation}\frac{d}n=f(n)+\frac{1+\Ge}n,\end{equation}}
{\begin{equation}\label{eq:13}nf(n)=d-1-\epsilon,\end{equation}}
where $\Ge=1-\frac sr=\frac{\{k\pmod{r}\}}{r}$ and hence $\frac{1}r\le\Ge\le 1$.

By \eqref{eq:12} and \eqref{eq:13}, we have
\begin{equation}~\label{n-k}
n-k=n-(1-f(n))\frac{r}{r+1}n=\frac n{r+1}+f(n)\frac{r}{r+1}n.\end{equation}
Let $H$ be a parity-check matrix of $\mC$ given in the form of \eqref{eq:8}. {Define $h=n-k-\ell$}, by \eqref{n-k} we have \[\ell+h=n-k=\frac n{r+1}+f(n)\frac{r}{r+1}n.\] By \eqref{eq:10},  we know that $\ell\ge \frac{n}{r+1}$. Thus, we may write
\begin{equation}~\label{gn}
\ell=\frac{n}{r+1}+g(n)\end{equation}
for some $g(n)$ with $0\le g(n)\le f(n)\frac{r}{r+1}n$, %This forces that $g(n)=o(n)$.
then
\[\frac{n}{r+1}+g(n)+h=\frac n{r+1}+f(n)\frac{r}{r+1}n.\]
Put $c(n)=f(n)\frac{r}{r+1}n-g(n)$,
so $h=c(n)\ge 0$. % and furthermore we have $c(n)=o(n)$.

Now we shall { {construct a $q$-ary locally repairable code}} with a specific form, so that we can apply Hamming bound to get the desired results.

Let $\mC$ be a $q$-ary $[n,k,d;r]$-locally repairable code with parity-check matrix $H$ in \eqref{eq:8}.
To start, we define two sets
\[A=\{i\in[n]:\; \mbox{the $i$-th column of $H_1$ has exactly one nonzero element}\}\]
and
\[B=\{i\in[n]:\; \mbox{the $i$-th column of $H_1$ has at least two nonzero elements}\}.\]
Put $a=|A|$ and $b=|B|$. Then we have
\begin{itemize}
\item[(i)] $a+b=n$;
\item[(ii)] $(r+1)\ell-b\ge n$.
\end{itemize}
Thus, by \eqref{gn} we have
\begin{equation}~\label{bn}\frac bn\le \frac{\ell(r+1)}n-1=(r+1)\times\frac{g(n)}{n}.\end{equation}
Through \eqref{bn}, we get
\[a=n-b \ge n-(r+1)g(n).\]

{Next, we construct aforementioned $q$-ary locally repairable code  in $3$ steps.}
\begin{itemize}
\item[Step $1.$] Denote by $L_1$ be the submatrix of $H$ obtained by deleting all columns indexed by $B$. By multiplying nonzero elements for each column of $L_1$, we may assume that $L_1$ has the following form
\begin{equation}\label{eq:14}
L_1=\left(
  \begin{array}{c}
    \begin{array}{c|c|ccc|c}
      {\bf 1} & \mathbf{0} & \cdots & \cdots & \cdots &\mathbf{0} \\
      \mathbf{0} & {\bf 1} & \cdots &\cdots &  \cdots & \mathbf{0} \\
      \vdots & \vdots & \ddots &\ddots & \ddots &\vdots \\
      \mathbf{0} & \mathbf{0} & \cdots & \cdots  & \cdots & {\bf 1}
    \end{array}
     \\
    \huge{L_2} \\
  \end{array}
\right),
\end{equation}
where each ${\bf 1}$ is the all-one vector of length at most $r+1$ and total number $\ell_1$ of all-one vectors are upper bounded by $\ell$; and $L_2$ is the $h\times a$ submatrix of $H_2$ indexed by $B$ (recall that $h=n-k-\ell$). It is clear that $L_1$ is an {$(n-k)\times a$} matrix.
\item[Step $2.$]
For each block of $L_1$, we do column operations by subtracting the first column of this block, we get a matrix of the following form
 \begin{equation}
L_3=\left(
    \begin{array}{c|c|ccc|c}
      1,{\bf 0} & \mathbf{0} & \cdots & \cdots & \cdots &\mathbf{0} \\
      \mathbf{0} & 1,{\bf 0} & \cdots &\cdots &  \cdots & \mathbf{0} \\
      \vdots & \vdots & \ddots &\ddots & \ddots &\vdots \\
      \mathbf{0} & \mathbf{0} & \cdots & \cdots  & \cdots & 1,{\bf 0}\\
      {\bf h}_{1},K_1&{\bf h}_2,K_2& \cdots & \cdots  & \cdots & {\bf h}_{\ell_1},K_{\ell_1}
    \end{array}
\right),
\end{equation}
where ${\bf h}_i$ is a column vector of length {${ h}$} and $K_i$ is a matrix of ${ h}$ rows and at most $r$ columns.
\item[Step $3.$] Put
\[K=(K_1,K_2,\dots,K_{\ell_1}).\]
Then $K$ has ${ h}$ rows and  $a-\ell_1$ columns. Since any $d-1$ columns of $H$ are linear independent, we have that any $\lfloor (d-1)/2\rfloor$ columns of $L_3$ are linearly independent.
\end{itemize}

Let ${\mC}_K$ be the linear code with $K$ as a parity-check matrix, then ${\mC}_K$ has length $a-\ell_1$, dimension at least $a-\ell_1-{ h}$ and distance at least $\lfloor (d-1)/2\rfloor+1$.

Note that by~(\ref{eq:13}), ${h}=c(n)=f(n)\frac{r}{r+1}n-g(n)\le f(n)\frac{r}{r+1}n=\frac r{r+1}(d-1-\Ge)$. Let $t=\{ d\pmod{4}\}$ and put $d_1=(d-t)/4$.

Case 1: $t=1$ or $2$.  Then $\lfloor (d-1)/2\rfloor+1=2d_1+1$. In this case, we get a $q$-ary $[a-\ell_1, \ge a-\ell_1-{ h},\ge 2d_1+1]$-linear code ${\mC}_K$. Applying the Hamming bound in \eqref{eq:3} to ${\mC}_K$ gives
\begin{equation*}
q^{a-\ell_1-{ h}}\le q^{\dim({\mC}_K)}\le\frac{q^{a-\ell_1}}{\sum_{i=0}^{d_1}{a-\ell_1\choose i}(q-1)^i}.
\end{equation*}
This gives ${a-\ell_1\choose d_1}(q-1)^{d_1}\le q^{ h}$. As  ${a-\ell_1\choose d_1}\ge \left(\frac{a-\ell_1}{d_1}\right)^{d_1}$, we get $\left(\frac{a-\ell_1}{d_1}\right)^{d_1}(q-1)^{d_1}\le q^{ h}$. Hence,
\begin{equation}\label{a-l} a-\ell\le a-\ell_1\le\frac{d_1}{q-1}q^{{h}/d_1}\le \frac{d-1}{4(q-1)}q^{4r(d-1-\Ge)/((d-t)(r+1))}.\end{equation}
By \eqref{eq:13} and $rd=o(n)$, we have $rf(n)=o(1)$. Thus, we have $rg(n)=o(n)$ as $\frac{rg(n)}{n}\le rf(n)\times\frac{r}{r+1}$.
Hence, we have
\[ a-\ell\ge n-(r+1)g(n)-\frac{n}{r+1}-g(n)=\frac{r}{r+1}\times n-(r+2)g(n)=\frac{r}{r+1}\times n+o(n).\]
 Then, we get
\[n\le (1+o(1))\frac{r+1}{r}\times\frac{d-1}{4(q-1)}q^{4r(d-1-\Ge)/((d-t)(r+1))}.\]

\iffalse
Note that $h=n-k-\ell\le n-k-\frac n{r+1}$, i.e.,
$h\le n-k-\left\lfloor\frac n{r+1}\right\rfloor$. Thus, if $d\le r+1$,  by \eqref{n-k} and \eqref{eq:13} we have
\begin{equation}\label{hx}
h\le n-k-\left\lfloor\frac n{r+1}\right\rfloor=\left\lfloor\frac{r}{r+1}\times nf(n)\right\rfloor= \left\lfloor\frac{(d-2)r+s}{r+1}\right\rfloor\le d-2.\end{equation}
Replacing $h$ in \eqref{a-l} by the upper bound on $h$ given in \eqref{hx}
\[n\le (1+o(1))\frac{r+1}{r}\times\frac{d-1}{4(q-1)}q^{4(d-2)/(d-t)}.\]
\fi
Case 2: $t=3$ or $4$. Then $\lfloor (d-1)/2\rfloor+1=2d_1+2$. In this case, we get a $q$-ary $[a-\ell_1, \ge a-\ell_1-{ h},\ge 2d_1+2]$-linear code ${\mC}_K$. By the propagation rule, there exists an $[a-\ell_1-1, q-\ell_1-{ h}, 2d_1+1]$-linear code.
In a similar way, applying the Hamming bound in \eqref{eq:3} to ${\mC}_K$ gives
\[n\le (1+o(1))\frac{r+1}{r}\times\frac{d-1}{4(q-1)}q^{4r(d-2-\Ge)/((d-t)(r+1))}.\]
\iffalse
With a similar argument as above, one can show that
\[n\le (1+o(1))\frac{r+1}{r}\times\frac{d-1}{4(q-1)}q^{4(d-3)/(d-t)}\]
if $d\le r+2$.\fi
This completes the proof.
\end{proof}
\begin{rmk}
 Note that the condition that $rd=o(n)$ in Theorem~\ref{thm:3.9} is used only at the very end of the proof, namely, when we argue that $a-\ell\ge\frac{nr}{r+1}+o(n)$ from the end of Case 1. In other words, in the proof of Theorem~\ref{thm:3.9}, we do not need this condition before discussing Case 1.
 \end{rmk}
 \begin{rmk}\label{rmk:5}\begin{itemize}
 \item[(i)] Under the condition that $(r+1)|n$ and recovery sets are disjoint, the paper \cite{GXY19} shows the following bound
  \begin{equation}\label{eq:rx}
  n\le \left\{\begin{array}{ll}
\frac{r+1}{r}\times\frac{d-1}{4(q-1)}q^{4(d-2)/(d-t))} &\mbox{if $d\equiv1, 2\pmod{4}$}\\
\frac{r+1}{r}\times\frac{d-1}{4(q-1)}q^{4r(d-3)/(d-t)} &\mbox{if $d\equiv3, 4\pmod{4}$},
\end{array}
\right.\end{equation}
where $t=\{ d\pmod{4}\}$.
  Our bound given in Theorem \ref{thm:3.9} is better than the one given in \eqref{eq:rx} if $d\ge r+2$, while the bound \eqref{eq:rx} outperforms the one in Theorem \ref{thm:3.9} if $d<r+2$.
  \item[(ii)] In fact, we can also derive the bound \eqref{eq:rx} without the condition that $(r+1)|n$ and recovery sets are disjoint. Note that $h=n-k-\ell\le n-k-\frac n{r+1}$, i.e.,
$h\le n-k-\left\lceil\frac n{r+1}\right\rceil$. Thus, if $d\le r+1$,  by \eqref{n-k} and \eqref{eq:13} we have
\begin{equation}\label{hx}
h\le n-k-\left\lceil\frac n{r+1}\right\rceil=\left\lfloor\frac{r}{r+1}\times nf(n)\right\rfloor= \left\lfloor\frac{(d-2)r+s}{r+1}\right\rfloor\le d-2.\end{equation} If $d\equiv1, 2\pmod{4}$,
replacing $h$ in \eqref{a-l} by the upper bound on $h$ given in \eqref{hx}, we get
\[n\le (1+o(1))\frac{r+1}{r}\times\frac{d-1}{4(q-1)}q^{4(d-2)/(d-t)}.\]
A similar argument works for the case where $d\equiv3, 4\pmod{4}$. In conclusion, we have the following result without assuming that $(r+1)|n$ and recovery sets are disjoint.
\begin{equation}\label{eq:yx}
  n\le \left\{\begin{array}{ll}
(1+o(1))\frac{r+1}{r}\times\frac{d-1}{4(q-1)}q^{4(d-2)/(d-t))} &\mbox{if $d\equiv1, 2\pmod{4}$}\\
(1+o(1))\frac{r+1}{r}\times\frac{d-1}{4(q-1)}q^{4r(d-3)/(d-t)} &\mbox{if $d\equiv3, 4\pmod{4}$},
\end{array}\right.\end{equation}
where $t=\{ d\pmod{4}\}$. The bound \eqref{eq:yx} coincides with \eqref{eq:rx}.
  \end{itemize}
  \end{rmk}

Next, we investigate upper bounds on the code lengths of Singleton-optimal locally repairable codes for small minimum distances such as $d=5,6$ and $7$. {Note that we only use some facts and results from the first part in the proof of Theorem \ref{thm:3.9}, thus we do not assume the condition that $(r+1)|n$ and recovery sets are disjoint for Theorems~\ref{thm:3.10}, \ref{thm:3.11} and \ref{thm:3.12}.}
\begin{theorem}~\label{thm:3.10}
If there is a Singleton-optimal $[n,k,d=5;r]$-locally repairable code, then we have
 \begin{itemize}
 \item[{\rm (i)}] $n=O(1)$ if $r=1$;
  \item[{\rm (ii)}] $n=O(q)$ if $r=2$ or $3$ and assume that the Main MDS conjecture holds;
    \item[{\rm (iii)}] $n=O(r)$ if $r\ge 4$, $k\equiv 0, -1$ or $-2\pmod{r}$;
    \item[{\rm (iv)}] $n=O(q^2)$ if $r\ge 4$ and $r=o(n)$.
 \end{itemize}
\end{theorem}

\begin{proof} {We follow the notations and approach given in the proof of Theorem~\ref{thm:3.9}.} %Thus, we think that $n$ is sufficiently large.
Let $k\equiv -s \pmod {r}, 0\le s \le r-1$.  Then, we have
\[5=n-k-\left\lceil\frac{k}r\right\rceil+2= n-k-\frac{k+s}r+2,\] i.e., \[k=\frac{r}{r+1}n-\frac{s}{r+1}-\frac{3r}{r+1}.\]
By ~\eqref{eq:10} and ~\eqref{eq:8}, we can obtain
\[\frac{n}{r+1}+h\le \ell+h=n-k=\frac{n}{r+1}+\frac{s+3r}{r+1},\] and
\begin{equation}\label{eq:16}
h\le \frac{s+3r}{r+1}=2+\frac{s+r-2}{r+1}.
\end{equation}
Since $0\le s\le r-1,$ so
\begin{equation}\label{eq:17}
\ell\le\frac{n}{r+1}+\frac{s+3r}{r+1}\le\frac{n}{r+1}+\frac{4r-1}{r+1}.
\end{equation}
Recall that $a+b=n$ and $(r+1)\ell-b\ge n$ in the proof of Theorem~\ref{thm:3.9}.  This gives $a\ge 2n-(r+1)\ell$.
Consider the code  ${\mC}_K$ with  length $a-\ell_1$ (refer to the proof of Theorem~\ref{thm:3.9} again), then we have
\begin{equation}\label{eq:l8} a-\ell_1\ge a-\ell\ge 2n-(r+2)\ell\ge 2n-(r+2)\left(\frac{n}{r+1}+\frac{4r-1}{r+1}\right).\end{equation}

{\it Case 1:} $r=1$. In this case, we have $s=0$. By \eqref{eq:l8}, the length of $\mC_K$ should be at least $\frac{n-9}{2}.$ {If $n<9$, then we have $n=O(1)$.} Now assume that $n\ge 9$. By \eqref{eq:16}, we have $h\le 1.5$, i.e, $h\le 1$. By \eqref{eq:l8}, $a-\ell_1\ge \frac{n-3}2\ge 3$. As ${\mC}_K$ gives an $[a-\ell_1,a-\ell_1-h, 3]$-linear code which is beyond the classical Singleton bound, this is a contradiction. The contradiction implies that $n<9$ if there exists a Singleton-optimal $[n,k,d=5;1]$-locally repairable code.

{\it Case 2:} $r=2$ or $3$. In this case, by \eqref{eq:16}, we have $h<3$, i.e., $h\le 2$. First of all, if $h=0$ or $1$, ${\mC}_K$ is linear code exceeding the classical Singleton bound. As we proved in Case 1, we must have $n=O(1)$. Now assume that $h=2$, then ${\mC}_K$ is an $[a-\ell_1,a-\ell_1-2, 3]$-linear code which achieves the classical Singleton bound. By the Main MDS conjecture~\cite{MS}, we have $a-\ell_1\le q+2$. Thus, we have
$2n-(r+2)\ell\le a-\ell\le a-\ell_1\le q+2$. Combining this with \eqref{eq:17}, we get
\[n\le \frac{r+1}{r}\times (q+2)+\frac{(r+2)(4r-1)}{r}=O(q).\]

Next we assume that $r\ge 4$.
The average length of each block of the matrix $L_1$ defined in \eqref{eq:14} is
\begin{equation*}\label{eq:lk}
\frac a{\ell_1}\ge\frac a{\ell} \ge \frac{2n}{\ell}-(r+1)\ge \frac{2n(r+1)}{n+4r-1}-(r+1).\end{equation*}
Thus, we have $\frac a{\ell_1}\ge\frac{2n(r+1)}{n+4r-1}-(r+1)\ge 4$
if $n\ge \frac{(r+5)(4r-1)}{r-3}$.

{\it Case 3:} $r\ge 4$ and  $k\equiv 0, -1$ or $-2\pmod{r}$, i.e., $0\le s\le 2$. In this case, we have $h<3$, i.e., $h\le 2$. Choose a block $M$ of the largest length within the matrix $L_1$ defined in \eqref{eq:14}, then the number of columns of $M$ is at least $4$. Furthermore, we know that any $4$ columns of $M$ are linearly independent. This is a contradiction since $M$ has at most three nonzero rows. The contradiction implies that $n< \frac{(r+5)(4r-1)}{r-3}=O(r)$.

{\it Case 4:} $r\ge 4$ and  $k\not\equiv 0, -1$ or $-2\pmod{r}$, i.e., $3\le s\le r-1$. In this case, we have $h<4$, i.e., $h\le 3$. If $h\le 2$, then we claim that $n=O(r)$ as proved in Case 3. This is impossible as $r=o(n)$. Now we assume that $h=3$. Then ${\mC}_K$ is a $[a-\ell_1,a-\ell_1-3,3]$-almost MDS code. By \cite{BB52}, we have $a-\ell_1\le q^2+q+1$.  Combining this with \eqref{eq:17}, we get \[n\le \frac{r+1}{r}\times (q^2+q+2)+\frac{(r+2)(4r-1)}{r+1}=O(q^2).\]
 This completes the proof.

\iffalse

$\ell\le\frac{n}{r+1}+\frac{a+3r}{r+1}\le\frac{n}{r+1}+\frac{4r-1}{r+1}$.
As in the proof of Theorem~\ref{thm:upperbound}, we have $a+b=n$ and $\ell(r+1)-b\ge n.$ Thus, $a\ge 2n-\ell(r+1).$  As a consequence, the average block length of $L_1$ is
\[\frac a{\ell_1}\ge \frac a{\ell}\ge \frac{2n}{\ell}-(r+1).\]

Let $G_4$ be a parity-check matrix of $\mC'.$ Now denote by $G_4$ be the submatrix of $G_3$ by choosing one block and deleting all zero rows and first column, we may get a matrix of the following form
\begin{equation}
G_4=\left(
  \begin{array}{c}
    \begin{array}{cccc}
       K_i
    \end{array}
     \\
  \end{array}
\right)_{h\times \left(\frac{a-\ell_1}{\ell_1}\right)}
\end{equation}
where $i\in\{1,2,\cdots, \ell_1\}.$ We have $\frac{a-\ell_1}{\ell_1}=\frac{a}{\ell_1}-1\ge \frac{a}{\ell}-1\ge \frac{2n}{\ell}-r\ge r+o(r).$
Since any $h$ columns of $G_4$ are linearly independent, so the minimum distance is at least $1+h$. Note that $t\equiv (d=5)\pmod{4}$ and $d_1=(d-t)/4=1$, then $d(\mC')=2d_1+1=3$.

If $h=2,$ $\mC'$ is an $q$-ary $[\frac{a-\ell_1}{\ell_1}, \frac{a-\ell_1}{\ell_1}-2, 3]$ MDS code. By the Main MDS conjecture, we have $q+2\ge\frac{a-\ell_1}{\ell_1}\ge r+o(r),$ i.e., $r=O(q).$ If $h=3,$ $\mC'$ is an $q$-ary $[\frac{a-\ell_1}{\ell_1}, \frac{a-\ell_1}{\ell_1}-3, 3]$ AMDS code. By \cite{BB52}, we have $ q^2+q+1\ge\frac{a-\ell_1}{\ell_1}\ge r+o(r),$ i.e., $r=O(q^2).$
\fi
\end{proof}

From part (i) of Theorem~\ref{thm:3.10}, there is  a Singleton-optimal $[n,k,d=5;r]$-locally repairable code with the length $n=O(1)$ if $r=1$. We provide an example as follows.
\begin{exm} Consider the code ${\mC}=\{\Gl(1,1,1,1,1):\;\Gl\in\F_q\}$. Then this gives a Singleton-optimal $[5,1,5;1]$-locally repairable code of length $5$.
\end{exm}
\iffalse
{\color{red}\begin{rmk} \begin{itemize}
\item[(i)] Note that the parameter regime $k\equiv 0, -1$ or $-2\pmod{r}$ is indeed possible. For instance, if $n-3$ is divisible $r+1$ and take $k=\frac{r(n-3)}{r+1}$, then we have $5=n-k-\left\lceil\frac kr\right\rceil+2$.
\item[(ii)]
Note that the upper bound  given in \cite{GXY19} is more restrictive comparing with Theorem~\ref{thm:upperbound}. More precisely speaking, under some parameter regime, the paper \cite{GXY19} gives an upper bound on length of locally repairable codes, while Theorem~\ref{thm:upperbound} shows an upper bound for under a less restrictive condition.  Futhermore, we improve the upper bound  given in \cite{GXY19} for some parameters. For instance, for
$d=5$ and $r|k$, we get an upper bound $n=O(dq^{2r/(r+1)})$, while the upper bound in \cite{GXY19} shows $n=O(dq^{2})$.
\end{itemize}
\end{rmk}}
\fi

\begin{rmk}
Note that the parameter constraint $k\equiv 0, -1$ or $-2\pmod{r}$ in part (iii) of Theorem~\ref{thm:3.10} is indeed possible. For instance, if $n-3$ is divisible $r+1$ and take $k=\frac{r(n-3)}{r+1}$, then we have $d=5=n-k-\left\lceil\frac kr\right\rceil+2$.
\end{rmk}

\begin{rmk}
Comparing Theorem~\ref{thm:3.10} with knwon results (see \cite{GXY19}, \cite{XY18} and \cite{CFXHF}), we get better upper bounds when $r=1,3$ and $r\ge 4$, $k\equiv 0, -1$ or $-2\pmod{r}.$ For the rest of the cases,  we obtain the same results with best known upper bounds.
\end{rmk}

\begin{theorem}\label{thm:3.11}
If there is a Singleton-optimal $[n,k,d=6;r]$-locally repairable code, then we have
 \begin{itemize}
 \item[{\rm (i)}] $n=O(q)$ if $r=1$ and the Main MDS conjecture holds;
  \item[{\rm (ii)}] $n=O(q)$ if $r=2$,  $k\equiv0\pmod{2}$ and the Main MDS conjecture holds; and $n=O(q^2)$ if $r=2$ and $k\equiv1\pmod{2}$;
    \item[{\rm (iii)}] $n=O(q^2)$ if $r=3$ or $4$;
    \item[{\rm (iv)}] $n=O(r)$ if $r\ge 5$,  $k\equiv 0, -1,-2$ or $-3\pmod{r}$;
    \item[{\rm (v)}] $n=O\left(q^{\frac{4r-2}{r+1}}\right)$ if $r\ge 5$ and $r=o(n)$.

 \end{itemize}
\end{theorem}

\begin{proof} The proof is quite similar to that of Theorem \ref{thm:3.10}. Let us sketch the proof only. Again, we follow the notations and approach given in the proof of Theorem~\ref{thm:3.9}.
Let $k\equiv -s \pmod {r}, 0\le s \le r-1$.  Then, we have
\[6=n-k-\left\lceil\frac{k}r\right\rceil+2= n-k-\frac{k+s}r+2,\] i.e., \[k=\frac{r}{r+1}n-\frac{s}{r+1}-\frac{4r}{r+1}.\]
By ~\eqref{eq:10} and ~\eqref{eq:8}, we can obtain \[\frac{n}{r+1}+h\le \ell+h=n-k=\frac{n}{r+1}+\frac{s+4r}{r+1},\] and
\begin{equation}\label{eq:19}
h\le \frac{s+4r}{r+1}=3+\frac{s+r-3}{r+1}.
\end{equation}So, we have
\begin{equation}\label{eq:20}
\ell\le\frac{n}{r+1}+\frac{s+4r}{r+1}\le\frac{n}{r+1}+\frac{5r-1}{r+1}.
\end{equation}

{\it Case 1:} $r=1$; or $r=2$ and $k\equiv0\pmod{2}$. In this case, by \eqref{eq:19}, we get $h\le 2$. Hence, ${\mC}_K$ is an MDS code. Under the Main MDS conjecture, we get $n=O(q)$.

{\it Case 2:} $r=2$ and $k\not\equiv0\pmod{2}$; or $r=3$; or $r=4$.  In this case, we have $h\le 3$. If $h\le 2$, then we have $n=O(q)$ as shown before. Now we assume that $h=3$. By the fact that ${\mC}_K$ is a $[a-\ell_1,a-\ell_1-3, 3]$-almost MDS code, we obtain that $n=O(q^2)$.

{\it Case 3:} $r\ge 5$ and  $k\equiv 0, -1,-2$ or $-3\pmod{r}$, i.e., $0\le s\le 3$.
\iffalse{\color{red}Assume that $r\ge 5$.
The average length of each block of the matrix $L_1$ defined in \eqref{eq:14} is
$
\frac a{\ell_1}\ge\frac a{\ell} \ge \frac{2n}{\ell}-(r+1)\ge \frac{2n(r+1)}{n+5r-1}-(r+1).$
Thus, we have $\frac a{\ell_1}\ge\frac{2n(r+1)}{n+5r-1}-(r+1)\ge 5$
if $n\ge \frac{(r+6)(5r-1)}{r-4}$.
}
\fi
Suppose that $n\ge \frac{(r+6)(5r-1)}{r-4}$.
In this case, we have $h\le 3$. The average length of each block of the matrix $L_1$ is at least $5$
under the condition that $n\ge \frac{(r+6)(5r-1)}{r-4}.$
Choose a block $M$ of the largest length within the matrix $L_1$ defined in \eqref{eq:14}, then the number of columns of $M$ is at least $5$. Furthermore, we know that any $5$ columns of $M$ are linearly independent. This is a contradiction since $M$ has at most four nonzero rows. The contradiction implies that $n<\frac{(r+6)(5r-1)}{r-4}=O(r)$.

{\it Case 4:}
The case where  $r\ge 5$ and  $k\not\equiv 0, -1,-2$ or $-3\pmod{r}$ follows from Theorem~\ref{thm:3.9}.
\end{proof}

\begin{rmk}
Comparing Theorem~\ref{thm:3.11} with knwon results (see \cite{GXY19}, \cite{XY18} and \cite{CFXHF}), we get better upper bounds when $r=2$ and  $k\equiv0\pmod{2}$, $r=4$ and $r\ge 5$ and  $k\equiv 0, -1,-2$ or $-3\pmod{r}.$ Note that in the case $r=2$ and  $k\equiv1\pmod{2},$ ~\cite{CFXHF} got the best upper bound, i.e., $n=O(q^{1.5}).$
 For the rest of the cases,  we obtain the same results with best known upper bounds.
\end{rmk}

\begin{theorem}\label{thm:3.12}
If there is  a Singleton-optimal $[n,k,d=7;r]$-locally repairable code, then we have
 \begin{itemize}
 \item[{\rm (i)}] $n=O(1)$ if $r=1$;
  \item[{\rm (ii)}] $n=O(q)$ if $r=2$ and the Main MDS conjecture holds;
  \item[{\rm (iii)}] $n=O(q)$ if $r=3$, $k\equiv 0\pmod{3}$ and the Main MDS conjecture holds; and $n=O(q^2)$ if $r=3$, $k\not\equiv 0\pmod{3}$;
    \item[{\rm (iv)}] $n=O(q^2)$ if $r=4$ or $5$;
    \item[{\rm (v)}] $n=O(r)$ if $r\ge 6$,  $k\equiv 0, -1,-2,-3$ or $-4\pmod{r}$;
    \item[{\rm (vi)}] $n=O\left(q^{\frac{4r-2}{r+1}}\right)$ if $r\ge 6$ and $r=o(n)$.

 \end{itemize}
\end{theorem}

\begin{proof} Again, the proof is quite similar to that of Theorem \ref{thm:3.10}. Let us sketch the proof only. Again, we follow the notations and approach given in the proof of Theorem~\ref{thm:3.9}.
Let $k\equiv -s \pmod {r}, 0\le s \le r-1$.  Then, we have
\[7=n-k-\left\lceil\frac{k}r\right\rceil+2= n-k-\frac{k+s}r+2,\] i.e., \[k=\frac{r}{r+1}n-\frac{s}{r+1}-\frac{5r}{r+1}.\]
By ~\eqref{eq:10} and ~\eqref{eq:8}, we can obtain
 \[\frac{n}{r+1}+h\le \ell+h=n-k=\frac{n}{r+1}+\frac{s+5r}{r+1},\] and
\begin{equation}\label{eq:21}
h\le \frac{s+5r}{r+1}=4+\frac{s+r-4}{r+1}.
\end{equation}
So, we have
\begin{equation}\label{eq:22}
\ell\le\frac{n}{r+1}+\frac{s+5r}{r+1}\le\frac{n}{r+1}+\frac{6r-1}{r+1}.
\end{equation}

{\it Case 1:} $r=1$. In this case, by \eqref{eq:21}, we get $h\le 2$. Hence, ${\mC}_K$ is a linear code exceeding the classical Singleton bound. Thus, $n=O(1)$.

{\it Case 2:} $r=2$ or $r=3$ and $k\equiv0\pmod{3}$.  In this case, we have $h\le 3$. If $h\le 2$, then we have $n=O(1)$ as shown before. Now we assume that $h=3$. By the fact that ${\mC}_K$ is a $[a-\ell_1,a-\ell_1-3, 4]$-MDS code, we obtain that $n=O(q)$.

{\it Case 3:} $r=3$ and  $k\not\equiv 0\pmod{3}$; or $r=4$ or $5$. In this case, we have $h\le 4$. If $h\le 3$, then we have $n=O(1)$ or $n=O(q)$ as shown before. Assume $h=4$. By the fact that ${\mC}_K$ is a $[a-\ell_1,a-\ell_1-4, 4]$-almost MDS code, we obtain that $n=O(q^2)$.

{\it Case 4:} $r\ge 6$ and  $k\equiv 0, -1,-2,-3$ or $-4\pmod{r}$, i.e., $0\le s\le 4$. In this case, we have $h\le 4$. The average length of  of each block of the matrix $L_1$ is at least $6$
if $n\ge \frac{(r+7)(6r-1)}{r-5}$. The last inequality is clearly satisfied as $r=o(n)$.
Choose a block $M$ of the largest length within the matrix $L_1$ defined in \eqref{eq:14}, then the number of columns of $M$ is at least $6$. Furthermore, we know that any $6$ columns of $M$ are linearly independent. This is a contradiction since $M$ has at most five nonzero rows. The contradiction implies that $n<\frac{(r+7)(6r-1)}{r-5}=O(r)$.

{\it Case 5:}
The last case follows from Theorem~\ref{thm:3.9} directly.
\end{proof}

\begin{rmk}
Assume that there is a Singleton-optimal $[n,k,d=8;r]$-locally repairable code. In the same way, we can prove (i)  $n=O(q)$ if $r=1$ and the Main MDS conjecture holds; and (ii) $n=O(q^2)$ if $r=2$; and (iii) $n=O(q^2)$ if $r=3$ and $k\not\equiv-2\pmod{3}$. For other cases, we refer to the upper bound on the length $n$ in Theorem~\ref{thm:3.9}.
\end{rmk}
Finally in this section, we show an upper bound for general minimum distance $d$ and locality $r$.
\begin{theorem}\label{thm:3.13}
Assume that there is a Singleton-optimal $[n,k,d;r]$-locally repairable code. If $s+2<d<r+2$, where $k\equiv -s\pmod{r}$ with $0\le s\le r-1$, then we have {$n< \frac{(d+r)(r-1+(d-2)r)}{r-d+2}$}. In particular, $n=O(1)$ if both $r$ and $d$ are constant.
\end{theorem}

\begin{proof} Suppose {$n\ge \frac{(d+r)(r-1+(d-2)r)}{r-d+2}$.}
We follow the notations and approach given in the proof of Theorem~\ref{thm:3.9}.
Let $k\equiv -s \pmod {r}, 0\le s \le r-1$.  Then, we have
\[d=n-k-\left\lceil\frac{k}r\right\rceil+2= n-k-\frac{k+s}r+2,\] i.e., \[k=\frac{r}{r+1}n-\frac{s}{r+1}-\frac{(d-2)r}{r+1}.\]
By ~\eqref{eq:10} and ~\eqref{eq:8}, we can obtain
 \[\frac{n}{r+1}+h\le \ell+h=n-k=\frac{n}{r+1}+\frac{s+(d-2)r}{r+1},\] and
\begin{equation}\label{eq:23}
h\le \frac{s+(d-2)r}{r+1}=d-3+\frac{s+r-d+3}{r+1}.
\end{equation}
So, we have
\begin{equation}\label{eq:24}
\ell\le\frac{n}{r+1}+\frac{s+(d-2)r}{r+1}.
\end{equation}
The average length of each block of the matrix $L_1$ defined in \eqref{eq:14} is
\[
\frac a{\ell_1}\ge\frac a{\ell} \ge \frac{2n}{\ell}-(r+1)\ge \frac{2n(r+1)}{n+s+(d-2)r}-(r+1)\ge d-1\]
{if $n\ge \frac{(d+r)(r-1+(d-2)r)}{r-d+2}$.}

Since $d< r+2$, we have $h\le d-3$ by \eqref{eq:23}. Choose a block $M$ of the largest length within the matrix $L_1$ defined in \eqref{eq:14}, then the number of columns of $M$ is at least $d-1$. Furthermore, we know that any $d-1$ columns of $M$ are linearly independent. This is a contradiction since $M$ has at most $d-2$ nonzero rows. This contradiction implies that $n< \frac{(d+r)(r-1+(d-2)r)}{r-d+2}.$ This proof is completed.
\end{proof}

\begin{rmk}
 If $k\equiv -(r-1)\pmod{r}$, then there are no such $d$ in the range $(s+2,r+2)$. This implies that Theorem \ref{thm:3.13} gives nothing. For other case, one can always find $d$ in the range $(s+2,r+2)$.
\end{rmk}

\begin{rmk}
It was shown in \cite{GXY19} that under the condition that $(r+1)|n$ and $d\le r+2$, a Singleton-optimal $[n,k,d;r]$-locally repairable code must have dimension $k\equiv 2-d\pmod{r}$, i.e., $d=s+2$. In this case, the length of a Singleton-optimal $[n,k,d;r]$-locally repairable code is $n=\Omega_{d,r}\left(q^{1+\frac1{\lfloor(d-3)/2\rfloor}}\right)$.
\end{rmk}

\section{Improved upper bounds with a constraint}
%{\color{red}$(r+1)|n$}
In \cite{GXY19} and \cite{CFXHF}, the authors studied upper bounds on length of Singleton-optimal locally repairable codes under the constraint that $(r+1)|n$ and recovery sets are disjoint. In this section, we assume this constraint as well to obtain  some improved upper bounds.

Our main idea is to introduce some propagation rules and then reduce the problem to determine upper bounds on lengths of Singleton-optimal locally repairable codes with small minimum distance. To make our propagation rules work well, we always assume that $n=\Omega(dr^2)$ in this section. With this assumption, we may assume that {$r^2+2r< n-d$.} Note that with assumption that $(r+1)|n$ and $n=\Omega(dr^2)$, it was proved in \cite{GXY19} that recovery sets are disjoint.

\begin{lemma} If $n=\Omega(dr^2)$, then a Singleton-optimal $[n,k,d;r]$-locally repairable code must have a recovery set of size $r+1$.
\end{lemma}
\begin{proof} Let ${\mC}$ be a Singleton-optimal $[n,k,d;r]$-locally repairable code. Suppose that every recovery set has size at most $r$.
If $r=1$, then $k=\frac{n-d}2+1\ge 1$. Since $r<k,$ the locality $r$ of ${\mC}$ is also $0$, this forces that ${\mC}$ is the trivial code $\{\bo\}$. This gives a contradiction.

Now assume that $r\ge 2$. Since ${\mC}$ is a Singleton-optimal $[n,k,d;r]$-locally repairable code, so \begin{equation}~\label{eq:s}
n-k-\left\lceil\frac{k}{r}\right\rceil+2=d.\end{equation}
In this case, the locality of $\mC$ is $r-1,$ {by the Singleton bound~\eqref{eq:1} and \eqref{eq:s}, we have
\[n-k-\left\lceil\frac{k}{r}\right\rceil+2=d\le n-k-\left\lceil\frac{k}{r-1}\right\rceil+2,\]}
this gives $\left\lceil\frac{k}{r}\right\rceil\ge \left\lceil\frac{k}{r-1}\right\rceil$.

Combining the inequality $2r+r^2<{n-d}$ and identity $n-k-\left\lceil\frac{k}{r}\right\rceil+2=d$ gives $k\ge r^2+r$.
Write $k=ur+v$ with $1\le v\le r$, so $u\ge r$. It can be rewritten as $k=u(r-1)+u+v$.
\iffalse
Combining the inequality $r<\sqrt{n-d}$ and identity $n-k-\left\lceil\frac{k}{r}\right\rceil+2=d$ gives $k\ge r^2+r$. Hence, $u\ge r$. \fi
Thus we have
\[\left\lceil\frac{k}{r}\right\rceil=u+1< u+2\le\left\lceil\frac{k}{r-1}\right\rceil.\]
This gives a contradiction and the proof is complete.
\end{proof}

\iffalse
Let $n_{\max}(d,q,r)$ be the maximum length of an $q$-ary optimal $(n,k,d;r)$-locally repairable code. In what follows, we only consider the case that
$d,r=O(\sqrt{n})$ and $n>q$. {\color{red}$d=\lambda n,\lambda<1.$} This implies that $k=\Omega(n)$ and thus $\lceil \frac{k}{r}\rceil<\lceil \frac{k}{r-1} \rceil$.
In other words, {a Singleton-optimal} $(n,k,d;r)$-locally repairable code is not {a Singleton-optimal} $(n,k,d;r-1)$-locally repairable code or equivalently, there always exists a recovery set of size $r+1$. We say that a $(n,k,d;r)$-locally repairable code code with disjoint recovery sets if there exists $\lceil \frac{n}{r+1}\rceil$ disjoint recovery sets covering all $n$ indices of this code.

%A $[n,k,d;r]$-locally repairable code is called near optimal if
%$$
%d=n-k-\lceil\frac{k}{r}\rceil+1.
%$$
\fi

Based on Lemma~\ref{lem:2.3}, let us have a propagation rule.
%Let us have a propagation rule first.
\begin{lemma}\label{lem:4.2}
Let $\mC$ be {a Singleton-optimal} $[n,k,d;r]$-locally repairable code over $\F_q$ with disjoint recovery sets. Then, there exists a Singleton-optimal
$[n-(r+1),k-a,d-r-1+a;r]$-locally repairable code  over $\F_q$ with disjoint recovery sets if $\left\lceil \frac{k}{r}\right\rceil=\left\lceil\frac{k-a}{r}\right\rceil$, where $0\le a\le r+1$ is an integer.
%Moreover, if $\lceil \frac{k}{r}\rceil=\lceil\frac{k-a}{r}\rceil+1$, then ${\mC}'$ is an near optimal locally repairable code, i.e.,
\end{lemma}
\begin{proof}
Since ${\mC}$ has disjoint recovery sets, without loss of generality, we may assume  $\{n,n-1,\ldots,n-r\}$ is a recovery set for indices $n-r,\ldots,n$. We first shorten code ${\mC}$ by removing the last $a$ columns from its parity-check matrix. Denote the resulting code by ${\mC}_1$. It is clear that ${\mC}_1=[n_1,k_1,d_1;r_1]$ has length $n_1=n-a$, dimension $k_1\ge k-a$, minimum distance $d_1\ge d$ and locality $r_1=r$. Furthermore, $\{n-r,n-r+1,\ldots,n-a\}$ is one of its disjoint recovery set. We then puncture the last $r+1-a$ indices from ${\mC}_1$ to obtain a linear code ${\mC}_2=[n_2,k_2,d_2;r_2]$. It is clear that ${\mC}_2$ still has locality $r_2=r$. Moreover, ${\mC}_2$ has code length $n_2=n-(r+1)$, dimension $k_2\geq k-a$ and minimum distance $d_2\geq d-(r+1)+a$. By taking a subcode of ${\mC}_2$, we may assume that $k_2=k-a$. Observe that
$$
d_2\geq d-(r+1)+a=(n-k-\left\lceil\frac{k}{r}\right\rceil+2)-(r+1)+a= n_2-k_2-\left\lceil \frac{k_2}{r}\right\rceil+2.$$
Note that in the last step, we use the fact that $\left\lceil \frac{k}{r}\right\rceil=\left\lceil\frac{k-a}{r}\right\rceil$.
%It is also clear that ${\mC}_1$ is near optimal if $\lceil \frac{k}{r}\rceil=\lceil\frac{k-a}{r}\rceil+1$.
\end{proof}

%\begin{rmk}
%We do not assume that $r+1|n$ in this theorem. In other words, this theorem can be applied to any code length of a locally repairable code code.
%\end{rmk}

In literature, there has been already some work on upper bounds on lengths of Singleton-optimal locally repairable codes with small distances \cite{GXY19} and \cite{CFXHF}. However, it is more difficult to upper bound lengths of Singleton-optimal locally repairable codes with large distances.
Lemma \ref{lem:4.2} provides a way to upper bound lengths of  Singleton-optimal locally repairable codes with larger distance via   Singleton-optimal locally repairable codes with smaller distances.

\begin{cor}\label{cor:4.3}
Let ${\mC}$ be {a Singleton-optimal} $[n,k,d;r]$-locally repairable code over $\F_q$ with disjoint recovery sets. If $d\equiv b \pmod{ r+1}$ and $n=\Omega(dr^2)$, then there exists {a Singleton-optimal}
$[n-d+b,k,b;r]$-locally repairable code over $\F_q$ with disjoint recovery sets.
\end{cor}
\begin{proof}
Applying Lemma \ref{lem:4.2} to ${\mC}$ by letting $a=0$ and we obtain {a Singleton-optimal} $[n-(r+1),k,d-(r+1);r]$-locally repairable code over $\F_q$ with disjoint recovery sets.
Note that there still exists a recovery set of size $r+1$ in ${\mC}$ as $r^2+2r<{n-d}$. After we do the same operations for $\frac{d-b}{r+1}$ times, we obtain the desired result.
\end{proof}

\subsection{The case that $r=1,2,3$ and $4$}

 Combining the results given in Section~\ref{sec:3} and the bound for the case $(r,d)=(2,6)$ given in \cite{CFXHF} with Corollary \ref{cor:4.3}, we obtain the following result.

\begin{theorem}\label{thm:r=2} Assume that $n=\Omega(dr^2)$.
Let ${\mC}$ be {a Singleton-optimal} $[n,k,d\geq 5;r]$-locally repairable code over $\F_q$ with disjoint recovery sets. Then, we have
\iffalse
{\small
\begin{center}
{\rm Table V\\ \bigskip
\begin{tabular}{ccccc}\hline\hline
&&$d=5\pmod{r+1}$&&\\ \hline
$r=1$&$r=2$&$r=3$&$r\ge 4$&$r\ge 4$ and $k\equiv0,-1,-2\pmod{r}$\\ \hline
$n=O(1)$&$n=O(q)$&$n=O(q)$&$n=O(q^2)$&--\\ \hline\hline
&&$d=6\pmod{r+1}$&&\\ \hline
$r=1$&$r=2$ &$r=2$ and $2|k$&$r=3$ or $4$&$r\ge 5$ and $k\equiv0,-1,-2,-3\pmod{r}$\\ \hline
$n=O(q)$&$n=O(q^{1.5})$&$n=O(q)$&$n=O(q^2)$&--\\ \hline\hline
&&$d=7\pmod{r+1}$&&\\ \hline
$r=1$&$r=2$ &$r=3$  &$r=3$ and $3|k$ &$r=4$ or $5$\\ \hline
$n=O(1)$&$n=O(q)$&$n=O(q^2)$&$n=O(q)$&$n=O(q^2)$\\ \hline\hline
\end{tabular}
}
\end{center}
}
``--" in the above table means that locally reparable codes with given parameters does not exist.
\fi
%%%%%%%%%%%%%%%
\begin{center}
{\rm Table V\\
Upper bounds  for $r=1,2,3,4$ with disjoint recovery sets\\ \smallskip
\setlength{\tabcolsep}{0.7mm}
\begin{tabular}
%{|c|c|c|c|c|c|}
{|m{2.5cm}<{\centering}|m{2.5cm}<{\centering}|m{2.5cm}<{\centering}|m{2.5cm}<{\centering}|m{4.8cm}<{\centering}|m{5.5cm}<{\centering}|}
\hline\hline
\multicolumn{5}{|c|}{$d=5\pmod{r+1}$}\\ \hline
$r=1$&$r=2$ & $r=3$&$r= 4$& {{$r\ge 4$ and $k\equiv0,-1,-2\pmod{r}$}}\\\hline
$n=O(1)$ &$n=O(q)$&$n=O(q)$&$n=O(q^2)$&--\\ \hline\hline
\multicolumn{5}{|c|}{$d=6\pmod{r+1}$}\\ \hline
$r=1$&$r=2$ &$r=2$ and $2|k$&$r=3$ or $4$&$r\ge 5$ and $k\equiv0,-1,-2,-3\pmod{r}$\\ \hline
$n=O(q)$&$n=O(q^{1.5})$&$n=O(q)$&$n=O(q^2)$&--\\ \hline\hline
\multicolumn{5}{|c|}{$d=7\pmod{r+1}$}\\ \hline
$r=1$&$r=2$ &$r=3$  &$r=3$ and $3|k$ &$r=4$ or $5$\\ \hline
$n=O(1)$&$n=O(q)$&$n=O(q^2)$&$n=O(q)$&$n=O(q^2)$\\ \hline\hline
\end{tabular}
}
\end{center}
``--" in the above table means that locally reparable codes with given parameters does not exist.
\end{theorem}

\iffalse
One can see that for $r=2$, Theorem \ref{thm:r=2} only covers the case $d=0,2 \bmod 3$. Thanks to Theorem \ref{thm:keythm}, we are able to prove that $n=O(q)$ as well when $r=2$ and $d=1 \bmod 3$.
\begin{theorem}\label{thm:smallr}  Assume that $r<\sqrt{n-2d}$.
Let ${\mC}$ be {a Singleton-optimal} $(n,k,d\geq 5;2)$-locally repairable code over $\F_q$ with disjoint recovery sets such that $r+1|n$ and $d=1\bmod 3$. Then, $n=O(q)$
\end{theorem}
\begin{proof}
We first show that this holds for $d=7$ and then generalize it to any $d=1\bmod 3$ by applying Corollary \ref{cor:reduction1}.
By the Singleton-type bound, {a Singleton-optimal} $(n,k,7;2)$-locally repairable code with $r+1|n$ obeys that
$$
k=\frac{rn}{r+1}-(d-2)+\lfloor \frac{d-2}{r+1}\rfloor=\frac{2n}{3}-5+1.
$$
Since $r+1|n$, $k=0 \mod 2$. This implies that $\lfloor \frac{k}{2}\rfloor=\lfloor \frac{k-1}{2}\rfloor$. Thus, by Theorem \ref{thm:keythm}, we obtain {a Singleton-optimal} $(n-3,k-1,d-2=5,2)$-locally repairable code code. By Theorem \ref{thm:r=2}, $n-3=O(q)$. Thus, we prove the case $d=7$. The proof is completed by applying Corollary \ref{cor:reduction1}.
\end{proof}

For $r=3$, Theorem 3 in \cite{XY18} already proves an upper bound $n=O(q^2)$ for all $d\geq 5$.
\fi

Based on Theorem~\ref{thm:r=2} and the bound for the case $(r,d)=(3,8)$, {$(r,d)=(4,8)$ and $(r,d)=(4,9)$} given in \cite{XY18}, we can summarize  the result for $1\le r\le 4$ in the following theorem.
\begin{theorem}\label{thm:4.5}
If there is {a Singleton-optimal} $[n,k,d\geq 5;r]$-locally repairable code over $\F_q$ with disjoint recovery sets, then we have
\begin{itemize}
\item[{\rm (i)}] $n=O(q)$ if $r=1$ with $2|d$ and the Main MDS conjecture holds; and $n=O(1)$ if $r=1$ with $2\nmid d$;
\item[{\rm (ii)}] (a) $n=O(q)$ if $r=2$ with $d\not\equiv 0\pmod{3}$ and the  Main MDS conjecture holds; (b) $n=O(q)$ if $r=2$ with $d\equiv 0\pmod{3}$ and $2|k$ and the  Main MDS conjecture holds; (c) $n=O(q^{1.5})$ if $r=2$ and $d\equiv 0\pmod{3}$;
\item[{\rm (iii)}] (a) $n=O(q^2)$ if $r=3$ and $d\equiv 0\pmod{4}$; (b) $n=O(q^2)$ if $r=3$, $d\equiv 0\pmod{4}$ and $k\not\equiv 1\pmod{3}$; (c)  $n=O(q)$ if $r=3$ and $d\equiv 1\pmod{4}$ and the Main MDS conjecture holds; (d) $n=O(q^2)$ if $r=3$ and $d\equiv 2,3\pmod{4}$; (e) $n=O(q)$ if $r=3$ and $d\equiv 3\pmod{4}$ with $3|k$ and the Main MDS conjecture holds;
\item[{\rm (iv)}] (a) $n=O(q^2)$ if $r=4$ and $d\equiv 0,1,2\pmod{5}$; (b) $n=O(q^2)$ if $r=4$ and $d\equiv 3,4\pmod{5}$.
\end{itemize}

\end{theorem}

\subsection{The case that $r>4$}

We now proceed to discuss the case where $r$ is relatively large. Note that in the previous subsection, we only assume that recovery sets are disjoint. In this subsection, we further assume that $r+1$ divides $n$.

\iffalse
The following theorem comes from \cite{GXY19}.
\begin{theorem}[see \cite{GXY19}]\label{thm:dividible}{\color{blue} Correct reference.}
{\color{red}Let ${\mC}$ be an optimal $[n,k,d\geq 5;r]$-locally repairable code over $\F_q$ with disjoint recovery sets such that $(r+1)|n$ and $d=4t+1$ with {\color{red}$t\ge 1$}. Then, $n=O(q^{3-\frac{1}{t}})$.}
\end{theorem}
This bound is better than the upper bound in \cite{XY18} except for $d=2,3,4,5 \pmod{ r+1}$.
\fi
Based on Lemma \ref{lem:4.2}, we are able to derive a stronger upper bound.
%To start with, we apply  Lemma \ref{lem:4.2} to prove an upper bound.

\begin{theorem}\label{thm:4.7}
Let ${\mC}$ be {a Singleton-optimal} $[n,k,d\geq r+2;r]$-locally repairable code over $\F_q$ with disjoint recovery sets such that $(r+1)|n$, $d\equiv1,2,3,4,5 \pmod {r+1}$. Then, $n=O(q^{2})$.
\end{theorem}
\begin{proof}
Note that our setting is $r>3$. The case $r\leq 3$ is proved in our previous subsection.
Let $d\equiv a \pmod {r+1}$ with $a=1,2,3,4,5$. By the Singleton-type bound in Lemma~\ref{lem:2.4}, we have
$$
k=\frac{rn}{r+1}-\left(d-2-\left\lfloor \frac{d-2}{r+1}\right\rfloor\right).
$$
This implies $k\equiv 0 \pmod{ r}$ when $a=1$ and $k\equiv r-a+2 \pmod{ r }$ otherwise. If $k\equiv 0 \pmod{ r}$, then $5-a<r$ and we have ${\lceil \frac{k}{r}\rceil}=\lceil \frac{k-(5-a)}{r} \rceil$. If $k\equiv r-a+2 \pmod{ r }$, then $r-a+2-(5-a)=r-3\geq 1$. This implies that
{ ${\lceil \frac{k}{r}\rceil}=\lceil \frac{k-(5-a)}{r} \rceil$.}
Therefore, by  Lemma \ref{lem:4.2} we obtain {a Singleton-optimal} $[n-(r+1),k-(5-a),d-(r+1)+5-a;r]$-locally repairable code ${\mC}$ over $\F_q$ with disjoint recovery sets and $(r+1)|n$. Since $d-(r+1)+5-a\equiv5 \pmod {r+1}$, by Corollary \ref{cor:4.3} and Proposition \ref{prop:1.1}, we have $n-(r+1)=O(q^2)$. The proof is completed.
\end{proof}

We now proceed to prove an upper bound for $d\pmod{ r+1}>5$. Note that this implicitly assumes $r+1>5$.
\begin{theorem}\label{thm:4.8}
Let ${\mC}$ be {a Singleton-optimal} $[n,k,d\geq r+2;r]$-locally repairable code over $\F_q$ with disjoint recovery sets such that $(r+1)|n$. If $d\equiv 4t-i \pmod{ r+1}$
with $1\le t\le r/4$ and $i\in \{-1,0,1,2\}$, then $n=O(q^{3-\frac{1}{t}})$.
\end{theorem}
\begin{proof}
The proof is almost the same as Theorem~\ref{thm:4.7}.
Let $d\equiv a \pmod {r+1}$ with $a=4t-i$. By the Singleton-type bound in Lemma~\ref{lem:2.4}, we have
$$
k=\frac{rn}{r+1}-\left(d-2-\left\lfloor \frac{d-2}{r+1}\right\rfloor\right).
$$
This implies $k\equiv r-a+2 \pmod r $ as $a\geq 2$.
{One can easily verify that $\lceil \frac{k}{r}\rceil=\lceil \frac{k-(i+1)}{r} \rceil$ as {$r-a+2-(i+1)=r+1-4t\geq 1$.}}
%{One can easily verify that $\lceil \frac{k}{r}\rceil=\lceil \frac{k-(i+1)}{r} \rceil$.}
Therefore, by  Lemma \ref{lem:4.2} we obtain {a Singleton-optimal} $[n-(r+1),k-(i+1),d-(r+1)+i+1;r]$-locally repairable code over $\F_q$ with disjoint recovery sets and $(r+1)|n$. Since $d-(r+1)+i+1\equiv 4t+1\pmod {r+1}$, by Corollary \ref{cor:4.3} and Proposition \ref{prop:1.1}, the proof is completed.
\end{proof}
\begin{rmk}
Theorem \ref{thm:4.7} and Theorem \ref{thm:4.8} show that in many cases, the maximum length of the Singleton-optimal locally repairable code is strictly less than $O(q^3)$ when $(r+1)|n$.
%For example, if $r=4t+1$, then for any $d\geq r+2$, we always have $n\leq O(q^{3-1/t})$.
\end{rmk}
%\section{General construction of optimal LRC codes}\label{sec:3}
%We note that the case $d=2,3,4,5 \bmod r+1$ is already proved in \cite{XY18}.
%We prove this theorem in two steps. In the first step, we show that this theorem holds for $d=1 \bmod r+1$. In the second step, we reduce the optimal $(n,k,d\geq r+1:r)$-LRC to the optimal $(n-(r+1),k',d':r)$-LRC with $d'=1\bmod r+1$. This will completes the proof.
%Let us proceed to the first step. Assume that ${\mC}_1$ be {a Singleton-optimal} $(n,k,d\geq r+1:r)$-LRC over $\F_q$ with $d=1 \bmod r+1$ and $r+1|n$.
%Set $d=a(r+1)+1$.
%By the Singleton-type bound, we have
%$$
%k=\frac{rn}{r+1}-(d-2)+\lfloor\frac{d-2}{r+1}\rfloor=-(a(r+1)-1)+a-1=0 \bmod r.
%$$
%Apply Theorem \ref{thm:keythm} to ${\mC}_1$ and we obtain {a Singleton-optimal} $(n-(r+1),k-2,d-(r-1):r)$-LRC code ${\mC}_2$ as $\lceil\frac{k}{r}\rceil=\lceil\frac{k-2}{r}\rceil$. Observe that $d-(r-1)=3 \bmod r+1$. This implies that $n-(r+1)=O(q^2)$. The first step is done. We proceed to the second step. Since $d=2,3,4,5 \bmod r+1$ is already proved in \cite{XY18}, it suffices to consider that
%$d=a(r+1)+b$ with $b>5$.
%By the Singleton-type bound, we have
%$$
%k=\frac{rn}{r+1}-(d-2)+\lfloor\frac{d-2}{r+1}\rfloor=-(a(r+1)+b-2)+a=r+2-b \bmod r.
%$$
%Apply Theorem \ref{thm:keythm} to ${\mC}_1$ and we obtain {a Singleton-optimal} $(n-(r+1),k-(r+1-b),d-b:r)$-LRC code ${\mC}_2$ as $\lceil\frac{k}{r}\rceil=\lceil\frac{k-(r+2-b)}{r}\rceil$.

\end{document}